\documentclass[10pt,journal,epsfig]{IEEEtran}

\usepackage[dvips]{graphicx}
\usepackage{graphicx}
\usepackage{amssymb}
\usepackage{cite}
\usepackage{subfigure}
\usepackage{amsmath}

\usepackage{algorithm}
\usepackage{algorithmic}

\begin{document}


\title{Fast Beam Alignment for Millimeter Wave Communications: A Sparse Encoding and Phaseless Decoding Approach}

\author{Xingjian Li, Jun Fang, Huiping Duan, Zhi Chen, and Hongbin Li, ~\IEEEmembership{Senior Member,~IEEE}
\thanks{Xingjian Li, Jun Fang and Zhi Chen are with the National Key Laboratory
of Science and Technology on Communications, University of
Electronic Science and Technology of China, Chengdu 611731, China,
Email: JunFang@uestc.edu.cn}
\thanks{Huiping Duan is with the School of Electronic Engineering,
University of Electronic Science and Technology of China, Chengdu
611731, China, Email: huipingduan@uestc.edu.cn}
\thanks{Hongbin Li is with the Department of Electrical and Computer Engineering,
Stevens Institute of Technology, Hoboken, NJ 07030, USA, E-mail:
Hongbin.Li@stevens.edu}
\thanks{This work was supported in part by the National Science
Foundation of China under Grant 61522104.}}

\maketitle


\begin{abstract}
In this paper, we studied the problem of beam alignment for
millimeter wave (mmWave) communications, in which we assume a
hybrid analog and digital beamforming structure is employed at the
transmitter (i.e. base station), and an omni-directional antenna
or an antenna array is used at the receiver (i.e. user). By
exploiting the sparse scattering nature of mmWave channels, the
beam alignment problem is formulated as a sparse encoding and
phaseless decoding problem. More specifically, the problem of
interest involves finding a sparse sensing matrix and an efficient
recovery algorithm to recover the support and magnitude of the
sparse signal from compressive phaseless measurements. A sparse
bipartite graph coding (SBG-Coding) algorithm is developed for
sparse encoding and phaseless decoding. Our theoretical analysis
shows that, in the noiseless case, our proposed algorithm can
perfectly recover the support and magnitude of the sparse signal
with probability exceeding a pre-specified value from
$\mathcal{O}(K^2)$ measurements, where $K$ is the number of
nonzero entries of the sparse signal. The proposed algorithm has a
simple decoding procedure which is computationally efficient and
noise-robust. Simulation results show that our proposed method
renders a reliable beam alignment in the low and moderate
signal-to-noise ratio (SNR) regimes and presents a clear
performance advantage over existing methods.
\end{abstract}




\section{Introduction}
Millimeter wave (mmWave) communication is a promising technology
for future cellular networks
\cite{RappaportMurdock11,RanganRappaport14,ChenZhao14,ChenSun16}.
It has the potential to offer gigabits-per-second communication
data rates by exploiting the large bandwidth available at mmWave
frequencies. Nevertheless, communication at the mmWave frequency
bands suffers from high attenuation and signal absorption
\cite{SwindlehurstAyanoglu14}. To address this issue, large
antenna arrays should be used to provide sufficient beamforming
gain for mmWave communications \cite{AlkhateebMo14}. In fact,
thanks to the small wavelength at the mmWave frequencies, the
antenna size is very small and thus a large number of array
elements can be packed into a small area, which makes the use of
large antenna arrays a feasible option for mmWave communications.


On the other hand, although directional beamforming helps
compensate for the significant path loss incurred by mmWave
signals, it comes up with a complicated beamforming training
procedure because, due to the narrow beam of the antenna array,
communication between the transmitter and the receiver is possible
only when the transmitter's and receiver's beams are well-aligned,
i.e. the beam directions are pointing towards each other.
Therefore, beamforming training is required to find the best
beamformer-combiner pair that gives the highest beamforming gain
\cite{LiuLi18}. A natural approach to perform beamforming training
is to exhaustively search for all possible beam pairs to identify
the best beam alignment, which requires the receiver to scan the
entire space for each choice of beam direction on the transmitter
side. This exhaustive search has a sample complexity of
$\mathcal{O}(N^2)$ ($N$ denotes the number of possible beam
directions) and usually takes a long time (up to several seconds)
to converge, particularly when the number of antennas at the
transmitter and the receiver is large \cite{AbariHassanieh16}.









To address this issue, many efforts have been made to reduce the
time required for beamforming training. Specifically, the IEEE
802.11ad standard proposed to conduct an exhaustive search at the
receiver, with the transmitter adopting a quasi-omnidirectional
beam pattern. This process is then reversed to have the
transmitter sequentially scan the entire space while the receiver
uses a quasi-omnidirectional beam shape. This protocol reduces
sample complexity from $\mathcal{O}(N^2)$ to $\mathcal{O}(N)$. To
further reduce the training time, adaptive beam alignment
algorithms, e.g.
\cite{SongChoi15,HurKim13,AlkhateebAyach14,NohZoltowski17,HussainMichelusi17},
were proposed. In these works, a hierarchical multi-resolution
beamforming codebook set is employed to avoid the costly
exhaustive sampling of all pairs of transmit and receive beams.
The basic idea is to use coarse codebooks to first identify the
range of the beam direction, and then use high-resolution
subcodebooks to find a finer beam direction. This adaptive beam
alignment requires to adaptively choose a subcodebook at each
stage based on the output of earlier stages, which requires
feedback from the receiver to the transmitter and may not be
available at the initial channel acquisition stage.



In addition to the above beam steering techniques, another
approach
\cite{AlkhateebyLeus15,SchniterSayeed14,KimLove15,GaoZhang15,
MarziRamasamy16,GaoDai15b,GaoDai16,ZhouFang16,ZhouFang17,LiFang18}
directly estimates the mmWave channel or its associated
parameters, e.g. angles of arrival/departure, without the need of
scanning the entire space. The rationale behind this class of
methods is to exploit the sparse scattering nature of mmWave
channels and formulate the channel estimation into a compressed
sensing problem. Although having the potential to substantially
reduce the training overhead, this compressed sensing-based
approach suffers from several drawbacks. Firstly, compressed
sensing methods usually involve a computational complexity that
might be too excessive for practical systems. Secondly, compressed
sensing methods require the knowledge of the phase of the
measurements. While in mmWave communications, due to the carrier
frequency offset (CFO) caused by high-frequency hardware
imperfections, the phase of the measurements might be corrupted by
a random noise that varies across time, and as a result, only the
magnitude information of the measurements is useful for beam
alignment. Lastly, for compressed sensing methods, the
beamforming/combining vectors have to be chosen to be random
vectors to satisfy the restricted isometry property. This,
however, comes at the cost of worse signal-to-noise ratio (SNR)
and reduced transmission range. Recently, a novel beam steering
scheme called as ``Agile-Link''
\cite{AbariHassanieh16,HassaniehAbari18} was proposed to find the
correct beam alignment. The proposed algorithm only uses the
magnitude information of the measurements for recovery of the
signal directions and achieves a sample complexity of
$\mathcal{O}(K\log N)$, where $K$ denotes the number of signal
paths.







In this paper, we continue the efforts towards developing a fast
and efficient beam alignment scheme for mmWave communications.
Similar to \cite{AbariHassanieh16,HassaniehAbari18}, we rely on
the magnitude information of the measurements for beam steering.
By exploiting the sparse scattering nature of mmWave channels, we
show that the beam alignment problem can be formulated as a sparse
encoding and phaseless decoding problem. More specifically, the
problem of interest is to devise a sparse sensing matrix
$\boldsymbol{A}$ (referred to as sparse encoding) and develop a
fast and efficient recovery algorithm (referred to as phaseless
decoding) to recover the support and magnitude information of the
sparse signal $\boldsymbol{x}$ from compressive phaseless
measurements:
\begin{align}
  \boldsymbol{y} = \left|\boldsymbol{A}\boldsymbol{x}\right|
  \label{observation-model}
\end{align}
Note that the estimation of sparse signals from compressive
phaseless measurements, termed as ``compressive phase retrieval
(CPR)'', has been extensively studied over the past few years,
e.g.
\cite{MoravecRomberg07,OhlssonYang12,ShechtmanBeck14,BahmaniRomberg15}.
Nevertheless, there are two important distinctions between our
problem and the standard CPR problem. First, standard CPR assumes
a random measurement matrix which satisfies the restricted
isometry property. For our problem, the measurement matrix which
determines the shape of the beam pattern cannot be designed
freely. In fact, to provide a sufficient beamforming gain for
signal reception, we need to impose a sparse structure on the
measurement matrix. Second, standard CPR aims to retrieve the
complete information of the sparse signal $\boldsymbol{x}$, while
for the beam alignment purpose, only partial information of
$\boldsymbol{x}$, i.e. the support and the magnitude information
of those nonzero entries, needs to recovered.



To our best knowledge, in existing literature, PhaseCode proposed
in \cite{PedarsaniYin17} is a CPR algorithm that is most relevant
to our sparse encoding and phaseless decoding problem, in which
its measurement matrix is devised based on a sparse-graph coding
framework. It was shown that PhaseCode can recover a $K$-sparse
signal using slightly more than $4K$ measurements with high
probability. Nevertheless, PhaseCode (even its robust version)
involves a delicate decoding procedure sensitive to noise and
measurement errors, and suffers from severe performance
degradation in the presence of noise. To overcome this difficulty,
in this work, we propose a sparse bipartite graph code (SBG-Code)
algorithm for sparse encoding and phaseless decoding. Different
from PhaseCode, our proposed method uses a set of sparse bipartite
graphs, instead of a single bipartite graph, to encode the sparse
signal. The proposed algorithm involves a simple decoding
procedure which has a minimum computational complexity and is
robust against noise. Also, it can recover the support and
magnitude information of a $K$-sparse signal with a sample
complexity of $\mathcal{O}(K^2)$, thus providing a competitive
solution for practical mmWave beam alignment systems.

The rest of the paper is organized as follows. In Section
\ref{sec:system-model}, the system model is discussed and the beam
alignment is formulated into a sparse encoding and phaseless
decoding problem. In Section \ref{sec:review}, an overview of
PhaseCode is provided. A SBG-Code method is developed in Section
\ref{sec:SBG-Code}, along with its theoretical analysis provided
in \ref{sec:theoretical-analysis}. The robust version of SBG-Code
is studied in \ref{sec:robust-SBG-Code} and the extension to
antenna array receiver is discussed in Section
\ref{sec:extension}. Simulation results are provided in Section
\ref{sec:simulation-results}, followed by concluding remarks in
Section \ref{sec:conclusions}.

\begin{figure}[!t]
\centering
\includegraphics[width=10cm]{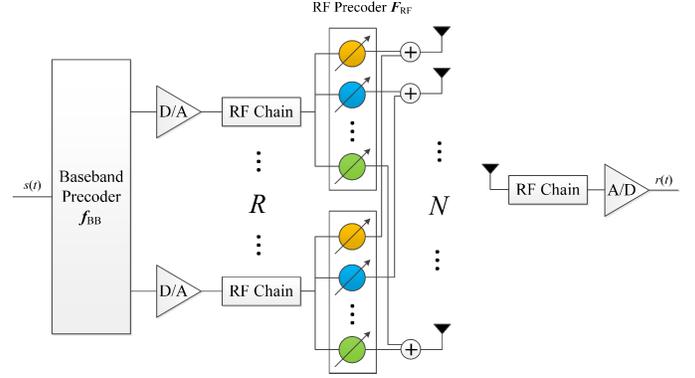}
\caption{The transmitter has a hybrid beamforming structure, and
the receiver uses an omni-directional antenna.} \label{fig6}
\end{figure}



\section{System Model} \label{sec:system-model}
Consider a mmWave communication system which consists of a
transmitter (base station) and a receiver (user). We assume that a
hybrid analog and digital beamforming structure is employed at the
transmitter, while the receiver has an omni-directional antenna
that receives in all directions (see Fig. \ref{fig6}). The
extension to an antenna array receiver will be discussed in
Section \ref{sec:extension}. The transmitter is equipped with $N$
antennas and $R$ RF chains. Since the RF chain is expensive and
power consuming, we have $R\ll N$. Note that although a single
receiver is considered in our paper, the extension of our scheme
to the multi-user scenario is straightforward, in which case the
base station periodically broadcasts a common codeword that is
decoded by each user to extract its associated channel information
\cite{SongHaghighatshoar18}. Each user then sends the index of the
beam corresponding to the selected angle-of-departure (AoD) to the
base station via a random access control channel. A connection
between the base station and the user is established after the
user receives a response from the base station.

The mmWave channel is characterized by a geometric channel model
\cite{AlkhateebAyach14}
\begin{align}
\boldsymbol{h} =
\sum_{p=1}^{P}\alpha_{p}\boldsymbol{a}_{t}(\theta_{p})
\label{channel-model}
\end{align}
where $P$ is the number of paths, $\alpha_{p}$ is the complex gain
associated with the $p$th path, $\theta_{p}\in[0,2\pi]$ is the
associated azimuth angle of departure (AoD), and
$\boldsymbol{a}_{t}\in\mathbb{C}^{N}$ is the transmitter array
response vector. Suppose a uniform linear array (ULA) is used.
Then the steering vector at the transmitter can be written as
\begin{align}
\boldsymbol{a}_{t}(\theta_{p})=\frac{1}{\sqrt{N}}
\left[1,e^{j\frac{2\pi}{\lambda}d\sin(\theta_{p})},\ldots,e^{j(N-1)
\frac{2\pi}{\lambda}d\sin(\theta_{p})}\right]^T
\end{align}
where $\lambda$ is the signal wavelength, and $d$ is the distance
between neighboring antenna elements. Due to the sparse scattering
nature of mmWave channels, $\boldsymbol{h}$ has a sparse
representation in the beam space (angle) domain:
\begin{align}
\boldsymbol{h}=\boldsymbol{D}\boldsymbol{x} \label{sparse-channel}
\end{align}
where $\boldsymbol{D}\in\mathbb{C}^{N\times N}$ is the discrete
Fourier transform (DFT) matrix, and
$\boldsymbol{x}\in\mathbb{C}^{N}$ is a $K$-sparse vector. If the
true AoA parameters $\{\theta_p\}$ lie on the discretized grid
specified by the DFT matrix, then the number of nonzero entries in
the beam space domain equals the number of signal paths, i.e.
$K=P$. The objective of beam alignment is to estimate the AoD and
the attenuation (in magnitude) of each path, which is equivalent
to recover the location indices and the magnitudes of the nonzero
entries in $\boldsymbol{x}$. The AoD of the dominant path is then
reported back to the base station via a control channel for beam
alignment.

Suppose the transmitter sends a constant signal $s(t)=1$ to the
receiver. The signal received at the $t$th time instant can be
expressed as
\begin{align}
r(t)=\boldsymbol{h}^T\boldsymbol{b}(t)s(t)+w(t)
=\boldsymbol{x}^T\boldsymbol{D}^T\boldsymbol{b}(t)+w(t)
\label{eqn11}
\end{align}
where $\boldsymbol{b}(t)\in\mathbb{C}^{N}$ is the
precoding/beamforming vector used by the transmitter at the $t$th
time instant, and $w(t)$ denotes the additive complex Gaussian
noise with zero mean and variance $\sigma^2$. Since a hybrid
analog and digital beamforming structure is employed at the
transmitter, the precoding vector can be expressed as
\begin{align}
\boldsymbol{b}(t)=\boldsymbol{F}_{\text{RF}}(t)\boldsymbol{f}_{\text{BB}}(t)
\label{eqn12}
\end{align}
in which $\boldsymbol{F}_{\text{RF}}(t)\in\mathbb{C}^{N\times R}$
and $\boldsymbol{f}_{\text{BB}}(t)\in\mathbb{C}^{R}$ represent the
radio frequency (RF) precoding matrix and the baseband (BB)
precoding vector, respectively. Specifically, to provide a
sufficient beamforming gain for signal reception, the transmitter
needs to form multiple beams simultaneously and steers them
towards different directions. To this objective, the RF precoding
matrix is chosen to be a submatrix of the complex conjugate of the
DTF matrix, $\boldsymbol{D}^{\ast}$
\begin{align}
\boldsymbol{F}_{\text{RF}}(t)=\boldsymbol{D}^{\ast}\boldsymbol{S}(t)
\label{eqn14}
\end{align}
where $\boldsymbol{S}(t)\in\mathbb{R}^{N\times R}$ is a column
selection matrix containing only one nonzero entry per column.
Note that each column of the DFT matrix can be considered as a
beamforming vector steering a beam to a certain direction. Hence,
the RF precoding matrix defined in (\ref{eqn14}) forms $R$ beams
towards different directions simultaneously.

Substituting (\ref{eqn12})--(\ref{eqn14}) into (\ref{eqn11}), we
obtain
\begin{align}
r(t)=\boldsymbol{x}^T\boldsymbol{a}(t)+w(t)=\boldsymbol{a}^T(t)\boldsymbol{x}+w(t)
\label{eqn6}
\end{align}
where
$\boldsymbol{a}(t)\triangleq\boldsymbol{S}(t)\boldsymbol{f}_{\text{BB}}(t)$
is an $N$-dimensional sparse vector with at most $R$ nonzero
elements. It should be noted (\ref{eqn6}) is an ideal model
without taking the CFO effect into account. In mmWave
communications, CFO is a factor that cannot be neglected, and, due
to the CFO between the transmitter and the receiver, the
measurements $r(t)$ will incur an additional unknown phase shift
that varies across time \cite{AbariHassanieh16}. Correcting this
unknown phase shift is difficult due to the high frequencies of
mmWave signals. In this case, only the magnitude information of
the measurements $r(t), t=1,\ldots,T$ is reliable.

Our objective is to devise a measurement matrix
$\boldsymbol{A}\triangleq
[\boldsymbol{a}(1)\phantom{0}\ldots\phantom{0}\boldsymbol{a}(T)]^T\in\mathbb{C}^{T\times
N}$ (referred to as sparse encoding) and develop a fast and
efficient recovery algorithm (referred to as phaseless decoding)
to recover $\boldsymbol{z}=|\boldsymbol{x}|$, i.e. the support and
magnitude of the sparse signal $\boldsymbol{x}$, from compressive
phaseless measurements:
\begin{align}
\boldsymbol{y}\triangleq
|\boldsymbol{r}|=|\boldsymbol{A}\boldsymbol{x}+\boldsymbol{w}|
\end{align}
where $\boldsymbol{r}\triangleq
[r(1)\phantom{0}\ldots\phantom{0}r(T)]^T$, and
$\boldsymbol{w}\triangleq
[w(1)\phantom{0}\ldots\phantom{0}w(T)]^T$. Note that the
measurement matrix $\boldsymbol{A}$ cannot be designed freely. As
discussed earlier, the transmitter has to form directional beams
for signal reception, otherwise the power of the signal may be too
weak to be received. To meet such a requirement, a sparse
structure is placed on $\boldsymbol{A}$:
\begin{itemize}
\item[C1] $\boldsymbol{A}$ is a sparse matrix with each row of
$\boldsymbol{A}$ containing at most $R$ nonzero elements.
\end{itemize}
For this reason, the design of the measurement matrix
$\boldsymbol{A}$ is referred to as sparse encoding. Also, since
the amount of time for beamforming training is proportional to the
number of measurement $T$, we wish $\boldsymbol{A}$ is properly
devised such that a reliable estimate of
$\boldsymbol{z}=|\boldsymbol{x}|$ can be obtained by using as few
measurements as possible.





\section{Review of Existing Solutions} \label{sec:review}
PhaseCode \cite{PedarsaniYin17} is a CPR algorithm that is most
relevant to our sparse encoding and phaseless decoding problem.
Here we first provide a brief review on PhaseCode. PhaseCode is an
efficient algorithm developed in a sparse-graph coding framework.
It consists of an encoding step and a decoding step. In the
encoding step, the measurement matrix
$\boldsymbol{A}\in\mathbb{C}^{4M\times N}$ is devised according to
\begin{align}
\boldsymbol{A}\triangleq \boldsymbol{H}\odot\boldsymbol{\bar{T}}
\end{align}
where $\odot$ denotes the Khatri-Rao product,
$\boldsymbol{H}\in\{0,1\}^{M\times N}$ is a binary code matrix
constructed using a random bipartite graph $G$ with $N$ left nodes
and $M$ right nodes, with its $(i,j)$th entry $H(i,j)=1$ if and
only if left node $j$ is connected to right node $i$, otherwise
$H(i,j)=0$. $\boldsymbol{\bar{T}}\in\mathbb{C}^{4\times N}$ is the
so-called ``trignometric modulation'' matrix that provides $4$
measurements for each row of $\boldsymbol{H}$, and
$\boldsymbol{\bar{T}}$ is given by
\begin{align}
\boldsymbol{\bar{T}}\triangleq \left[
\begin{array}{cccc}
    e^{j\omega} & e^{j2\omega} & \cdots & e^{jN\omega}  \\
    e^{-j\omega} & e^{-j2\omega} & \cdots & e^{-jN\omega} \\
    2\cos(\omega) & 2\cos(2\omega) & \cdots & 2\cos(N\omega) \\
    e^{j\omega'} & e^{j2\omega'} & \cdots & e^{jN\omega'} \\
  \end{array}
\right]
\end{align}
where $\omega\in (0,2\pi/N]$, and $\omega'$ is a random phase
between $0$ and $2\pi$. In the decoding stage, a delicate
procedure is employed to recover $\boldsymbol{x}$. It was shown in
\cite{PedarsaniYin17} that, in the noiseless case, PhaseCode can
recover a $K$-sparse signal with high probability using only
slightly more than $4K$ measurements. This theoretical result
suggests that the sample complexity required for beam alignment
can be significantly reduced to as low as $\mathcal{O}(K)$.
Nevertheless, there are two major issues when applying PhaseCode
to the beam alignment problem. Firstly, in PhaseCode, the
bipartite graph $G$ used to determine the binary code matrix
$\boldsymbol{H}$ is randomly generated. There is no guarantee that
the resulting measurement matrix $\boldsymbol{A}$ satisfies
constraint C1. Secondly, PhaseCode involves a delicate decoding
procedure requiring a high accuracy of the measurements, and
suffers from severe performance degradation in the presence of
noise. This makes PhaseCode an unsuitable solution for beam
alignment problems where measurements are inevitably contaminated
by noise.








\begin{figure}[!t]
\centering
\includegraphics[width=8cm]{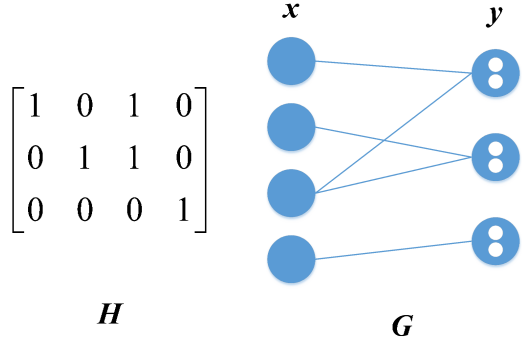}
\caption{The bipartite graph $G$ and its associated binary code
matrix $\boldsymbol{H}$, in which each left node of $G$
corresponds to an component of $\boldsymbol{x}$, and each right
node of $G$ corresponds to the set of measurements obtained via
the corresponding row of $\boldsymbol{H}$.} \label{fig7}
\end{figure}


\section{Proposed SBG-Coding Algorithm} \label{sec:SBG-Code}
To overcome the drawbacks of existing solutions, we propose a
sparse bipartite graph-Code (SBG-Code) algorithm for sparse
encoding and phaseless decoding.


\subsection{Sparse Encoding}
Different from PhaseCode, the proposed SBG-Code uses a set of
bipartite graphs $\{G_l\}_{l=1}^L$, instead of a single bipartite
graph, to encode the sparse signal. Let
$\boldsymbol{H}_l\in\{0,1\}^{M\times N}$ denote the binary code
matrix associated with the graph $G_l$ with $N$ left nodes and $M$
right nodes. The $(i,j)$th entry of $\boldsymbol{H}_l$ is given by
\begin{align}
H_{l}(i,j)= \begin{cases} 1 & \text{if and only if left node $j$
of $G_l$ is connected}\\
& \text{to right node $i$ of $G_l$} \\
0 & \text{otherwise}  \end{cases}
\end{align}
Given $\{\boldsymbol{H}_l\}$, the measurement matrix
$\boldsymbol{A}\in\mathbb{R}^{2ML\times N}$ is devised as
\begin{align}
\boldsymbol{A}\triangleq \left[
\begin{array}{c}
    \boldsymbol{H}_1 \odot \boldsymbol{T}   \\
    \boldsymbol{H}_2 \odot \boldsymbol{T} \\
    \vdots \\
    \boldsymbol{H}_L \odot \boldsymbol{T} \\
  \end{array}
\right] \label{A-design}
\end{align}
where $\boldsymbol{T}\in\mathbb{R}^{2\times N}$ is a simplified
trignometric modulation matrix defined as
\begin{align}
\boldsymbol{T}\triangleq \left[
  \begin{array}{cccc}
    1 & 1 & \cdots & 1  \\
    2\cos(\omega) & 2\cos(2\omega) & \cdots & 2\cos(N\omega) \\
  \end{array}
\right] \label{modulation-matrix}
\end{align}
in which $\omega\in(0,\pi/(2N)]$ such that $\cos(\omega
l)\in[0,1)$. We will show later the trignometric function
$\cos(n\omega)$ can be replaced by a general function.


For each graph $G_l$, each of its left node can be deemed as a
component of the sparse signal $\boldsymbol{x}$, and each right
node of $G_l$ refers to a set of $2$ measurements obtained as (see
Fig. \ref{fig7})
\begin{align}
\boldsymbol{y}_{l,m}=|(\boldsymbol{H}_l[m,:]\odot\boldsymbol{T})\boldsymbol{x}|
\quad \forall m=1,\ldots, M
\end{align}
where $\boldsymbol{H}_l[m,:]$ denotes the $m$th row of
$\boldsymbol{H}_l$. A left node, say node $n$, is called as active
left node if the $n$th signal component, $x_n$, is nonzero. For a
$K$-sparse signal $\boldsymbol{x}$, there are $K$ active left
nodes in total. A right node is called as a nullton, a singleton
or a multiton if:
\begin{itemize}
\item Nullton: A right node is a nullton if it is not connected
to any active left node.
\item Singleton: A right node is a singleton if it is connected
to exactly one active left node.
\item Multiton: A right node is a multiton if it is connected to
more than one active left node.
\end{itemize}
A bipartite graph which does not contain any multiton right nodes
is called as
\begin{itemize}
\item No-Multiton-graph (NM-graph): A bipartite
graph whose right nodes are either singletons or nulltons.
\end{itemize}
For our proposed SBG-Code, the purpose of employing multiple
bipartite graphs is to ensure that, with overwhelming probability,
there exists at least an NM-graph, i.e. a bipartite graph whose
right nodes are either singletons or nulltons.

The bipartite graphs $\{G_l\}$ with $N$ left nodes and $M$ ($M>K$)
right nodes are designed as follows. First, for simplicity, we
assume $r\triangleq N/M$ to be an integer. For each graph, we
randomly divide $N$ left nodes into $M$ equal-size, disjoint sets
(i.e. each set has $r$ left nodes) and establish a one-to-one
correspondence between $M$ sets of left nodes and $M$ right nodes.
If $N$ is not an integer multiple of $M$, we can still divide $N$
left nodes into $M$ disjoint sets, with all sets, except the last
one, consisting of $r=\text{floor}(N/M)$ left nodes. Clearly, a
right node is a singleton (nullton) if its corresponding set of
left nodes contains only one (zero) active left node. As to be
shown later, such design helps maximize the probability that a
bipartite graph is an NM-graph, i.e. its rights nodes are either
singletons or nulltons. Clearly, for each bipartite graph $G_l$
devised as described, its corresponding binary code matrix
$\boldsymbol{H}_l$ has only one nonzero element per column, and at
most $r$ nonzero elements per row. As a result, each row of the
resulting measurement matrix $\boldsymbol{A}$ contains at most $r$
nonzero elements. We can therefore choose $r\leq R$, which is
equivalent to $M\geq N/R$, such that $\boldsymbol{A}$ satisfies
constraint C1. Once $\boldsymbol{A}$ is given, the RF precoding
matrices $\{\boldsymbol{F}_{\text{RF}}(t)\}$ and baseband
precoding vectors $\{\boldsymbol{f}_{\text{BB}}(t)\}$ can be
accordingly determined.





\subsection{Phaseless Decoding}
Next, we discuss how to retrieve the support and magnitude
information of $\boldsymbol{x}$ from compressive phaseless
measurement $\boldsymbol{y}$. We first ignore the observation
noise in order to simplify our exposition and analysis, i.e.
\begin{align}
\boldsymbol{y}=|\boldsymbol{A}\boldsymbol{x}| \label{CPR-problem}
\end{align}
Let
\begin{align}
\boldsymbol{A}_l\triangleq\boldsymbol{H}_l \odot \boldsymbol{T}
\label{Al-definition}
\end{align}
denote the $l$th measurement sub-matrix associated with the
bipartite graph $G_l$, and
\begin{align}
\boldsymbol{y}_l\triangleq|\boldsymbol{A}_l\boldsymbol{x}|
\end{align}
denote the corresponding measurements. Suppose $G_l$ is an
NM-graph. If a right node is a nullton, it does not connect to any
active left nodes and thus we have
$\boldsymbol{y}_{l,m}=\boldsymbol{0}$. Therefore we only need to
consider those singleton right nodes. A singleton right node means
that only one nonzero component of $\boldsymbol{x}$, say $x_n$,
contributes to the value of $\boldsymbol{y}_{l,m}$. More
precisely, we can write
\begin{align}
\boldsymbol{y}_{l,m}=\left[
  \begin{array}{c}
    |x_n|  \\
    |2\cos(n\omega)x_n| \\
  \end{array}
\right]
\end{align}
Clearly, the magnitude and location index of $x_n$ can be readily
estimated as
\begin{align}
z_{\hat{n}}=& y_{l,m}^{(1)} \nonumber\\
\hat{n}=&
\frac{1}{\omega}\arccos\left(\frac{y_{l,m}^{(2)}}{2y_{l,m}^{(1)}}\right)
\label{eqn1}
\end{align}
where $y_{l,m}^{(1)}$ and $y_{l,m}^{(2)}$ denote the first and
second entry of $\boldsymbol{y}_{l,m}$, respectively. Note that
the graph $G_l$ is designed such that each right node is
exclusively connected to a subset of left nodes, and every left
node belongs to a certain subset that is connected to a certain
right node. Therefore, by performing the estimation (\ref{eqn1})
for all singleton right nodes, we are guaranteed to find the
location indices and magnitudes of all active left nodes. From the
above discussion, we see that if a bipartite graph, say graph
$G_l$, is an NM-graph, then $\boldsymbol{z}=|\boldsymbol{x}|$ can
be recovered from the corresponding phaseless measurements
$\boldsymbol{y}_l$.



The problem is that since we do not have the support information
of the sparse signal in advance, there is no guarantee that a
designed graph is an NM-graph which only contains singleton and
nullton right nodes. To address this issue, we employ multiple
bipartite graphs to encode the sparse signal, with the hope that
there exists at least one NM-graph. Note that in our algorithm, we
do not need to know which bipartite graph is an NM-graph. We just
perform the decoding as if the graph is an NM-graph, even if this
may not be true. To see this, suppose the graph $G_l$ is not an
NM-graph and contains a multiton. The multiton right node is a
superposition of multiple active left nodes, say, $x_{n_1}$ and
$x_{n_2}$, i.e.
\begin{align}
\boldsymbol{y}_{l,m}=\left[
  \begin{array}{c}
    |x_{n_1}+x_{n_2}|  \\
    |2(\cos(n_1\omega)x_{n_1}+\cos(n_2\omega)x_{n_2})| \\
  \end{array}
\right]
\end{align}
Clearly, performing (\ref{eqn1}) by treating
$\boldsymbol{y}_{l,m}$ as a singleton yields incorrect location
index and magnitude information. Nevertheless, in this case, it is
clear that the estimate of $\boldsymbol{z}=|\boldsymbol{x}|$ based
on $\boldsymbol{y}_{l}$, denoted as $\boldsymbol{\hat{z}}^{(l)}$,
contains less than $K$ nonzero components. This is an important
observation based on which we can differentiate the correct
estimate from incorrect estimates. Due to the fact that $K$ is
unknown in practice, given the $L$ estimates
$\{\boldsymbol{\hat{z}}^{(l)}\}_{l=1}^L$, the final estimate can
be chosen to be the one which has the largest number of nonzero
entries. Obviously, our proposed algorithm succeeds to recover the
support and magnitude of the sparse signal as long as there exists
at least one NM-graph. For clarity, our proposed algorithm is
summarized in Algorithm \ref{algorithm1}.

\begin{algorithm}[!t]
\caption{Proposed SBG-Code Algorithm}
\begin{algorithmic}
\STATE {Given
$\boldsymbol{A}_l=\boldsymbol{H}_l\odot\boldsymbol{\tilde{T}}$ and
$\boldsymbol{y}_l$ for each bipartite graph $G_l$, $l=1,\ldots,L$
\FOR{$l=1,\ldots,L$} \FOR{$m=1,\ldots,M$} \IF
{$\boldsymbol{y}_{l,m} \neq \boldsymbol{0}$} \STATE Assume the
$m$th right node is a singleton. \STATE Estimate the magnitude and
the location index of the active left node connected to the $m$th
right node via (\ref{eqn1}) \ENDIF \ENDFOR \STATE Obtain an
estimate of $\boldsymbol{z}$, denoted as
$\hat{\boldsymbol{z}}^{(l)}$. \ENDFOR \STATE Given the $L$
estimates $\{\boldsymbol{\hat{z}}^{(l)}\}_{l=1}^L$, choose the
estimate that has the largest number of nonzero entries as the
final estimate.}
\end{algorithmic}
\label{algorithm1}
\end{algorithm}

We see that, through the use of multiple bipartite graphs, the
proposed SBG-Code circumvents the complicated decoding procedure
that is needed by PhaseCode to check whether a right node is a
singleton, a mergeable multiton or a resolvable multiton. Although
the use of multiple bipartite graphs could bring a higher sample
complexity, the simplified decoding procedure can help improve the
robustness against measurement errors and noise.

\subsection{Discussions}
It should be noted that the cosine function used in
(\ref{modulation-matrix}) to encode the sparse signal can be
replaced by a general function. For example, a linear function
$f(n)=n/N$ can be employed to encode the sparse signal, in which
case the trignometric modulation matrix $\boldsymbol{T}$ is
expressed as
\begin{align}
\boldsymbol{T}= \left[
  \begin{array}{cccc}
    1 & 1 & \cdots & 1  \\
    1/N & 2/N & \cdots & 1 \\
  \end{array}
\right] \label{modulation-matrix-2}
\end{align}
Correspondingly, the $m$th singleton right node can be written as
\begin{align}
\boldsymbol{y}_{l,m}=\left[
  \begin{array}{c}
    |x_n|  \\
    |n x_n/N| \\
  \end{array}
\right]
\end{align}
and the magnitude and location index of $x_n$ can be readily
estimated as
\begin{align}
z_{\hat{n}}=&y_{l,m}^{(1)} \nonumber\\
\hat{n}=& \frac{N y_{l,m}^{(2)}}{y_{l,m}^{(1)}} \label{estimator2}
\end{align}






\section{Theoretical Analysis for SBG-Code} \label{sec:theoretical-analysis}
We now provide theoretical guarantees for our proposed SBG-Code
scheme. We first analyze the probability of a bipartite graph
being an NM-graph. To simplify our analysis, we assume
$r\triangleq N/M$ is an integer. The results are summarized as
follows.

\subsection{Main Results}
\newtheorem{proposition}{Proposition}
\begin{proposition}
Suppose we have
\begin{align}
\boldsymbol{y}_l=|\boldsymbol{A}_l\boldsymbol{x}| \label{eqn2}
\end{align}
where $\boldsymbol{x}\in\mathbb{C}^{N}$ is a $K$-sparse signal,
and the location indexes of its nonzero components are chosen
uniformly at random. $\boldsymbol{A}_l$ is defined in
(\ref{Al-definition}), in which
$\boldsymbol{H}_l\in\{0,1\}^{M\times N}$ is a binary code matrix
constructed according to a given bipartite graph $G_l$.
Specifically, $\boldsymbol{H}_l$ (i.e. $G_l$) is designed such
that each column of $\boldsymbol{H}_l$ has at least one nonzero
element, and the $m$th row of $\boldsymbol{H}_l$ has $r_{m}$
nonzero elements. If $M\geq K$, then the probability that all
right nodes of $G_l$ are either singletons or nulltons is upper
bounded by
\begin{align}
P(\text{$G_l$ is an NM-graph}) \leq r^K C_M^K/
C_N^K\triangleq\lambda
  \label{inequality1}
\end{align}
where $C_N^K$ denotes the number of $K$-combinations from a set
with $N$ elements. Also, the inequality (\ref{inequality1})
becomes an equality if and only if
\begin{align}
r_{1}=\cdots=r_{M}=r
\end{align}
\label{proposition1}
\end{proposition}
\begin{proof}
See Appendix \ref{appA}.
\end{proof}

From Proposition \ref{proposition1}, we know that the probability
of a bipartite graph being an NM-graph is maximized when
$r_{m}=r,\forall m$, in which case each column of
$\boldsymbol{H}_l$ has only one nonzero element, and each row of
$\boldsymbol{H}_l$ has exactly $r$ nonzero elements. This result
explains why we devise the bipartite graphs $\{G_l\}$ as discussed
in Section \ref{sec:SBG-Code}.A. Based on this result, our
proposed phaseless decoding scheme can recover
$\boldsymbol{z}=|\boldsymbol{x}|$ from compressive phaseless
measurements with probability given as follows.

\newtheorem{theorem}{Theorem}
\begin{theorem}
Consider the phaseless decoding problem described in
(\ref{CPR-problem}), where the measurement matrix
$\boldsymbol{A}\in\mathbb{R}^{2ML\times N}$ is generated according
to our proposed sparse encoding scheme. If $M\geq K$, then our
proposed algorithm can recover $\boldsymbol{z}=|\boldsymbol{x}|$
from phaseless measurements (\ref{CPR-problem}) with probability
exceeding
\begin{align}
  p = 1-\left(1-\lambda\right)^L
  \label{success-rate}
\end{align}
where $\lambda$ is defined in (\ref{inequality1}).
\label{theorem1}
\end{theorem}
\begin{proof}
See Appendix \ref{appB}.
\end{proof}

Note that our proposed algorithm requires a total number of $T=
2ML$ phaseless measurements, in which $M$ is the number of right
nodes per bipartite graph and $L$ is the number of bipartite
graphs. From (\ref{inequality1}), we see that increasing $M$ helps
achieve a higher $\lambda$, which in turn leads to a higher
recovery probability for our algorithm. On the other hand,
increasing $L$ improves the probability that there exists at least
one NM-graph among those $L$ bipartite graphs, and thus can also
enhance the recovery probability. Therefore, given the total
number of measurements $T$ fixed, there is a tradeoff between the
choice of $M$ and $L$. Here we provide an example to illustrate
this tradeoff. Suppose $N=128$, $K=2$ and $T=32$. The parameters
$M$ and $L$ can be chosen as one of the following cases, and the
exact recovery probability of our proposed algorithm can be
accordingly calculated:
\begin{itemize}
  \item $M=16, L=1$: $p=94.4882\%$
  \item $M=8, L=2$: $p=98.6050\%$
  \item $M=4, L=4$: $p=99.6450\%$
  \item $M=2, L=8$: $p=99.6333\%$
\end{itemize}
From this example, we see that choosing a moderate value for $M$
and $L$ provides the best performance.


\subsection{Analysis of Sample Complexity}
Let $M=\delta K$, where $\delta>1$ is parameter whose choice will
be discussed later. From Theorem \ref{theorem1}, we can derive the
number of bipartite graphs required for perfectly recovering
$|\boldsymbol{x}|$ with probability exceeding a prescribed
threshold $p_0$:
\begin{align}
L\geq &
\frac{\log(1-p_0)}{\log(1-\lambda)}=\frac{\log[(1-p_0)^{-1}]}{\log[(1-\lambda)^{-1}]}
\end{align}
As a result, the total number of measurements required for exact
recovery with probability exceeding $p_0$ is given by
\begin{align}
T=2ML=2\delta KL\geq \frac{c\delta K }{\log[(1-\lambda)^{-1}]}
\label{number-measurements-1}
\end{align}
where $c\triangleq 2\log[(1-p_0)^{-1}]> 0$ is a constant
determined by $p_0$. Note that $\lambda$ defined in
(\ref{inequality1}) can be lower bounded by
\begin{align}
  \lambda &= \frac{M!}{M^K(M-K)!}\frac{N^K(N-K)!}{N!} \nonumber \\
  &\geq \frac{M!}{M^K(M-K)!} \nonumber \\
  &\geq \frac{(M-K+1)^K}{M^K} 
  =\left(1-\frac{1-K^{-1}}{\delta}\right)^K\triangleq f(K,\delta)
\end{align}
Define
\begin{align}
h(K,\delta) \triangleq \frac{1}{\log[(1-f(K,\delta))^{-1}]}
\end{align}
The term on the right-hand side of (\ref{number-measurements-1})
can be upper bounded by
\begin{align}
\frac{c\delta K }{\log[(1-\lambda)^{-1}]} \leq c\delta K
h(K,\delta) \label{eqn3}
\end{align}
To facilitate analyzing the sample complexity of our proposed
algorithm, we choose $\delta=K$, i.e. $M=K^2$, which is a choice
usually offering satisfactory performance. In this case, it can be
easily proved that the function $f(K,\delta)$ decreases with an
increasing $K$, and
\begin{align}
\lim_{K\rightarrow +\infty} f(K,\delta)\mid_{\delta=K} = e^{-1}
\label{eqn5}
\end{align}
Therefore $h(K,\delta)\mid_{\delta=K}$ can be upper bounded by
\begin{align}
h(K,\delta)\mid_{\delta=K}\leq \frac{1}{\log[(1-e^{-1})^{-1}]}
\approx 1.51 \label{eqn4}
\end{align}
Combining (\ref{eqn3}) and (\ref{eqn4}), we can reach that, when
$\delta=K$, the term on the right-hand side of
(\ref{number-measurements-1}) is upper bounded by
\begin{align}
\frac{c\delta K }{\log[(1-\lambda)^{-1}]}  \leq 1.51cK^2
\end{align}
In other words, if the total number of phaseless measurements $T$
satisfies
\begin{align}
 T\geq 1.51cK^2 \label{sample-complexity}
\end{align}
then our proposed algorithm can perfectly recover
$|\boldsymbol{x}|$ with probability exceeding $p_0$. From
(\ref{sample-complexity}), we see that the sample complexity for
our proposed algorithm is of order $\mathcal{O}(K^2)$, which,
surprisingly, is independent of the dimension of the sparse
signal, $N$. Such a result can be well explained because for the
typical choice of $\delta=K$, the probability of a bipartite graph
being an NM-graph is lower bounded by $e^{-1}$ (cf. (\ref{eqn5}))
even for an arbitrarily large $N$. But notice that the irrelevance
of the sample complexity to $N$ is achieved by increasing $r$
since we have $r=N/M$ and $M$ is kept fixed as $K^2$ as $N$ grows.
In the beam alignment application, $r$ cannot become arbitrarily
large due to the limited number of RF chains.

Although a typical choice of $M=K^2$ is adopted for analyzing the
sample complexity, it should not be difficult to reach a similar
conclusion for a general choice of $M$ with $M=\mathcal{O}(K^2)$.
As a comparison, note that the sample complexity attained by most
compressive phase retrieval methods
\cite{ShechtmanBeck14,BahmaniRomberg15} and the AgileLink beam
steering scheme \cite{AbariHassanieh16,HassaniehAbari18} is of
order $\mathcal{O}(K\log(N))$.




\section{Robust SBG-Code Algorithm} \label{sec:robust-SBG-Code}
The basic idea of our proposed GF-Code algorithm is to divide the
$N$ components of $\boldsymbol{x}$ (i.e. $N$ left nodes) into $M$
disjoint sets, and each set of left nodes is connected to an
individual right node. If a right node is a singleton, it means
that its corresponding set of left nodes contains only one active
left node whose location and magnitude can be easily estimated via
(\ref{eqn1}) or (\ref{estimator2}), depending on which modulation
matrix is used. Such an idea works perfectly for the noiseless
case. Nevertheless, when the measurements are corrupted by noise,
a perfect estimate of the magnitude of the active left node is
impossible. Besides, the location index of the active left node
may be incorrectly estimated as well. In the following, we develop
a robust scheme for sparse encoding and phaseless decoding in the
presence of noise.





\subsection{Robust Sparse Encoding}
To facilitate our following exposition, the trignometric
modulation matrix (\ref{modulation-matrix}) or
(\ref{modulation-matrix-2}) is expressed as a general form as
\begin{align}
\boldsymbol{T}\triangleq \left[
  \begin{array}{cccc}
    1 & 1 & \cdots & 1  \\
    t_1 & t_2 & \cdots & t_N \\
  \end{array}
\right] \label{general-modulation-matrix}
\end{align}
where $t_i\neq t_j$ for $i\neq j$, and $t_n>0,\forall
n=1,\ldots,N$. Let $\{m_1^{(l)},\ldots,m_r^{(l)}\}$ denote the set
of indices of the left nodes connected to the $m$th right node of
the graph $G_l$. Note that the index set
$\{m_1^{(l)},\ldots,m_r^{(l)}\}$ is determined once the
corresponding bipartite graph $G_l$, i.e. $\boldsymbol{H}_l$, is
given. Here we assume $r=N/M$ is an integer. The extension to the
non-integer case is straightforward, as discussed earlier in
Section \ref{sec:SBG-Code}. Also, for simplicity, the superscript
$l$ used to denote the index of the bipartite graph is omitted,
and in the following, $\{m_1^{(l)},\ldots,m_r^{(l)}\}$ is
simplified as $\{m_1,\ldots,m_r\}$.

Suppose the $m$th right node is a singleton and $x_{m_i}$ is the
active left node connected to the $m$th right node, in which
$m_i\in\{m_1,\ldots,m_r\}$. When noise is present, the
measurements corresponding to the $m$th right node of the graph
$G_l$ can be expressed as
\begin{align}
\boldsymbol{y}_{l,m}=\left[
  \begin{array}{c}
    |x_{m_i}+w_{l,m}^{(1)}|  \\
    \left|t_{m_i}x_{m_i}+w_{l,m}^{(2)}\right| \\
  \end{array}
\right] \triangleq \left[
  \begin{array}{c}
    y_{l,m}^{(1)}  \\
    y_{l,m}^{(2)} \\
  \end{array}
\right]
\end{align}
where $w_{l,m}^{(1)}$ and $w_{l,m}^{(2)}$ denote the observation
noise added to the first and the second entry of the $m$th right
node, respectively. In this case, the location index of the active
left node can be estimated as
\begin{align}
\hat{m_i}=&\ \mathop{\arg\min}_{m_i\in\{m_1,\ldots,m_r\}}
\left|t_{m_i}-\frac{y_{l,m}^{(2)}}{y_{l,m}^{(1)}}\right|
\label{location-estimator}
\end{align}
The problem lies in that, if the index set $\{m_1,\ldots,m_r\}$
contains an element $m_j$ such that $t_{m_j}$ is close to
$t_{m_i}$, then it is likely that the location index of the active
left node is misidentified as $m_j$ since when noise is present,
we may have
\begin{align}
\left|t_{m_i}-\frac{y_{l,m}^{(2)}}{y_{l,m}^{(1)}}\right|>
\left|t_{m_j}-\frac{y_{l,m}^{(2)}}{y_{l,m}^{(1)}}\right|
\end{align}
To improve robustness against noise, it is clear that the absolute
difference $|t_{m_i}-t_{m_j}|$ should be as large as possible for
any pair of indices $\{m_i,m_j\}$ chosen from the set
$\{m_1,\ldots,m_r\}$.

Inspired by this insight, we propose to use an individual
modulation matrix for each bipartite graph. Specifically, the
modulation matrix for each bipartite graph is a column-permuted
version of the original modulation matrix, i.e.
\begin{align}
\boldsymbol{T}_l=\boldsymbol{T}\boldsymbol{P}_l \quad \forall l
\label{robust-modulation-matrix}
\end{align}
where $\boldsymbol{T}_l$ denotes the modulation matrix for graph
$G_l$, and $\boldsymbol{P}_l$ is a permutation matrix to be
devised. Write
\begin{align}
\boldsymbol{T}_l\triangleq \left[
  \begin{array}{cccc}
    1 & 1 & \cdots & 1  \\
    t_1^{(l)} & t_2^{(l)} & \cdots & t_N^{(l)} \\
  \end{array}
\right]
\end{align}
Following a similar deduction, the location index of the active
left node associated with the $m$th right node can be estimated as
\begin{align}
\hat{m_i}=&\ \mathop{\arg\min}_{m_i\in\{m_1,\ldots,m_r\}}
\left|t_{m_i}^{(l)}-\frac{y_{l,m}^{(2)}}{y_{l,m}^{(1)}}\right|
\end{align}
Therefore, if the permutation matrix $\boldsymbol{P}_l$ is devised
such that for each right node $m$, the elements in the
corresponding set $\{t_{m_1}^{(l)},\ldots,t_{m_r}^{(l)}\}$ are
sufficiently separated, then the robustness against noise can be
improved. To put it in a mathematical way, define the pairwise
distance associated with the $m$th right node as
\begin{align}
d_{m}^{(l)}\triangleq \min_{1\leq i<j\leq r}
\left|t_{m_i}^{(l)}-t_{m_j}^{(l)}\right|
\end{align}
Then the design of $\boldsymbol{P}_l$ can be formulated as a
Max-Min problem whose objective is to maximize the minimum
pairwise distance among the pairwise distances associated with $M$
right nodes, i.e.
\begin{align}
\max_{\boldsymbol{P}_l}\min_{m} \quad d_{m}^{(l)}
\end{align}
Such an optimization can be solved by searching for all possible
permutation matrices. Note that it is more advantageous to use an
individual permutation matrix for each graph than using a common
permutation matrix for all graphs because employing an individual
permutation matrix for each graph can help achieve a larger
minimum pairwise distance.





\subsection{Robust Phaseless Decoding}
We next devise a robust decoding scheme to estimate
$\boldsymbol{z}=|\boldsymbol{x}|$ from noisy measurements
$\boldsymbol{y}$. In the noisy case, the measurements
$\boldsymbol{y}$ are written as
\begin{align}
\boldsymbol{y}=|\boldsymbol{A}\boldsymbol{x}+\boldsymbol{w}|
\end{align}
where the measurement matrix $\boldsymbol{A}$ is expressed as
\begin{align}
\boldsymbol{A}\triangleq \left[
\begin{array}{c}
    \boldsymbol{A}_1   \\
    \boldsymbol{A}_2 \\
    \vdots \\
    \boldsymbol{A}_L \\
  \end{array}
\right] \triangleq \left[
\begin{array}{c}
    \boldsymbol{H}_1 \odot \boldsymbol{T}_1   \\
    \boldsymbol{H}_2 \odot \boldsymbol{T}_2 \\
    \vdots \\
    \boldsymbol{H}_L \odot \boldsymbol{T}_L \\
  \end{array}
\right]
\end{align}
and the modulation matrix $\boldsymbol{T}_l$ for graph $G_l$ is
given by (\ref{robust-modulation-matrix}). The measurements
associated with the bipartite graph $G_l$ are give by
\begin{align}
  \boldsymbol{y}_l\triangleq|\boldsymbol{A}_l\boldsymbol{x}+\boldsymbol{w}_l|
\end{align}
and the measurements, $\boldsymbol{y}_{l,m}\in\mathbb{R}^{2}$,
corresponding to the $m$th right node of $G_l$ are expressed as
\begin{align}
\boldsymbol{y}_{l,m}=|(\boldsymbol{H}_{l}[m,:] \odot
\boldsymbol{T}_l)\boldsymbol{x}+\boldsymbol{w}_{l,m}| \quad
\forall m=1,\ldots, M
\end{align}
where $\boldsymbol{w}_{l,m}$ denotes the noise added to the $m$th
right node of $G_l$. Due to the presence of noise, we usually have
$\boldsymbol{y}_{l,m}\neq\boldsymbol{0}$ even if the $m$th right
node is a nullton. Hence we first need to decide whether a right
node of $G_l$ is a nullton or not. Such a problem can be
formulated as a binary hypothesis test problem:
\begin{align}
  &H_0:\ y_{l,m}^{(1)} = |w_{l,m}^{(1)}| \nonumber \\
  &H_1:\ y_{l,m}^{(1)} = \left|\sum_{m_i\in S} x_{m_i}+w_{l,m}^{(1)}\right|
\end{align}
where $w_{l,m}^{(1)}$ is the additive complex Gaussian noise with
zero mean and variance $\sigma^2$, and $S$ denotes the set of
indices of those active left nodes that are connected to the $m$th
right node. A simple energy detector can be used to perform the
detection:
\begin{align}
y_{l,m}^{(1)}\underset{H_0}{\overset{H_1}{\gtrless}}\epsilon
\label{energy-detector}
\end{align}
It is clear that $y_{l,m}^{(1)}$ under $H_0$ follows a Rayleigh
distribution. Given a prescribed false alarm probability, the
threshold $\epsilon>0$ can be easily determined from the
distribution of $y_{l,m}^{(1)}$ under $H_0$. Such an energy
detector is able to yield satisfactory detection performance for a
moderate and high signal-to-noise ratio.

To proceed with our decoding scheme, we assume all nullton right
nodes of $G_l$ are correctly identified. In this case, we are able
to determine whether $G_l$ is an NM-graph or not. Specifically, if
$G_l$ is an NM-graph, then it contains $M-K$ nullton right nodes;
otherwise the number of nullton right nodes is greater than $M-K$.
Although the number of active left nodes, $K$, is unknown \emph{a
priori}, those graphs which have the smallest number of nullton
right nodes can be considered as NM-graphs and $K$ can be simply
estimated as
\begin{align}
\hat{K}=M-J
\end{align}
where $J$ denotes the smallest number of nullton right nodes among
all graphs.

We now perform decoding on those NM-graphs. Suppose $G_l$ is an
NM-graph and its $m$th right node is a singleton. Also, $x_{m_i}$
is the active left node connected to the $m$th right node. From
the discussion in the previous subsection, it is clear that the
magnitude and location index of this active left node can be
estimated as
\begin{align}
z_{\hat{m}_i}=&\ y_{l,m}^{(1)}  \nonumber \\
\hat{m}_i=&\ \mathop{\arg\min}_{m_i\in\{m_1,\ldots,m_r\}}
\left|t_{m_i}^{(l)}-\frac{y_{l,m}^{(2)}}{y_{l,m}^{(1)}}\right|
\label{estimator3}
\end{align}
where $\{m_1,\ldots,m_r\}$ denotes the set of indices of those
left nodes connected to the $m$th right node. After performing
(\ref{estimator3}) for all singleton right nodes, we are able to
obtain an estimate of $\boldsymbol{z}=|\boldsymbol{x}|$. Let
$\boldsymbol{\hat{z}}^{(l)}$ denote an estimate of
$\boldsymbol{z}$ obtained from the measurements associated with
$G_l$. Since we may have more than one NM-graphs, we are able to
collect multiple estimates of $\boldsymbol{z}$. The problem lies
in, due to the existence of noise, these multiple estimates,
denoted as
$\{\boldsymbol{\hat{z}}^{(1)},\ldots,\boldsymbol{\hat{z}}^{(I)}\}$,
are not exactly the same. In the following, we propose a
set-intersection scheme to combine these multiple estimates into a
more accurate estimate.






To better illustrate our idea, suppose there are two NM-graphs,
say $G_i$ and $G_j$, and $x_n$ is the only active left node in
$\boldsymbol{x}$. Recall that for each bipartite graph, the $N$
left nodes are divided into $M$ disjoint sets, with each set of
left nodes connected to an individual right node. Let
$S^{(i)}_{n}$ denote the set of left nodes to which $x_n$ belongs
in graph $G_i$, and $S^{(j)}_{n}$ denote the set of left nodes to
which $x_n$ belongs in graph $G_j$. Suppose the singleton right
nodes in both $G_i$ and $G_j$ are correctly identified. Then we
know that $x_n$ belongs to both $S^{(i)}_{n}$ and $S^{(j)}_{n}$.
If the intersection of the two sets $S^{(i)}_{n}$ and
$S^{(j)}_{n}$, $S^{(i)}_{n}\cap S^{(j)}_{n}$, contains only one
element, then it must be $x_n$ and the location of $x_n$ can be
uniquely determined. Such an idea can be easily extended to the
scenario where there are more then two NM-graphs, and for such a
case, the set-intersection scheme is more likely to succeed
because the more sets are used, the higher the probability of the
intersection of these sets containing only one element.

There, however, is a problem for the general case where
$\boldsymbol{x}$ contains multiple nonzero components (i.e.
multiple active left nodes). In this case, we have no idea which
set of left nodes a certain active node belongs to for each
NM-graph. As a result, it is impossible to determine which sets
should be put together to perform the intersection operation. To
overcome this difficulty, we note that the magnitudes of those
active left nodes are generally different. Hence the estimated
magnitude can be used to identify a certain active left node.
Without loss of generality, let $x_1,\ldots,x_K$ denote the
nonzero components of $\boldsymbol{x}$ in decreasing order in
terms of magnitude, i.e. $|x_1|> \cdots
> |x_K|>0$. For each NM-graph, say graph $G_i$, we can obtain an estimate of
$|\boldsymbol{x}|$, denoted as $\boldsymbol{z}^{(i)}$.
Specifically, let $\hat{z}_{i_1}>\cdots>\hat{z}_{i_K}>0$ represent
the nonzero components of $\boldsymbol{\hat{z}}^{(i)}$, then the
$k$th largest element $\hat{z}_{i_k}$ can be regarded as an
estimate of $|x_k|$. For each NM-graph, say $G_i$, the set of left
nodes containing $x_k$ can therefore be determined as the set of
left nodes containing $\hat{z}_{i_k}$. A set intersection
operation can then be performed to yield the final estimate of the
location index of $x_k$. On the other hand, the magnitude of the
$k$th largest component of $\boldsymbol{x}$ can be estimated as
the average of all estimates, i.e.
\begin{align}
|\hat{x}_{k}|=\frac{1}{I}\sum_{i=1}^{I}\hat{z}_{i_k}
\end{align}
Note that if the intersection of the sets contains more than one
element, then we randomly pick up an element in the intersection
set as the estimate of the location index of $x_k$. In addition,
in case the intersection is an empty set, which is possible due to
the incorrect association between $\{x_1,\ldots,x_K\}$ and
$\{\hat{z}_{i_1},\ldots,\hat{z}_{i_K}\}$, we randomly select an
estimate from
$\{\boldsymbol{\hat{z}}^{(1)},\ldots,\boldsymbol{\hat{z}}^{(I)}\}$
as the final estimate. For clarity, our proposed robust SBG-Code
algorithm is summarized in Algorithm \ref{algorithm2}. We see the
proposed decoding algorithm involves very simple addition and
multiplication calculations, and thus is amiable for practical
implementation.

\begin{algorithm}[!h]
\caption{Robust SBG-Code Algorithm}
\begin{algorithmic}
\STATE {Given
$\boldsymbol{A}_l=\boldsymbol{H}_l\odot\boldsymbol{\tilde{T}}_l$
and $\boldsymbol{y}_l$ for each bipartite graph $G_l$,
$l=1,\ldots,L$ \FOR{$l=1,\ldots,L$} \STATE Decide whether a right
node of $G_l$ is a nullton or not via the energy detector
(\ref{energy-detector}). Count the number of nulltons of $G_l$.
\ENDFOR \STATE Find graphs that have the smallest number of
nulltons and consider them as NM-graphs \FOR{$l=1,\ldots,L$}
\IF{$G_l$ is an NM-graph} \FOR {$m=1,\ldots,M$} \IF
{$y_{l,m}^{(1)}>\epsilon$} \STATE Assume the $m$th right node is a
singleton \STATE Estimate the magnitude and the location index of
the active left node connected to the $m$th right node via
(\ref{estimator3}) \ENDIF \ENDFOR \STATE Obtain the estimate
$\hat{\boldsymbol{z}}^{(l)}$ \ENDIF \ENDFOR \STATE Given multiple
estimates $\{\boldsymbol{\hat{z}}^{(l)}\}_{l=1}^I$ \IF{$I=1$}
\STATE Choose $\boldsymbol{\hat{z}}^{(1)}$ as the final estimate
\ELSE \STATE Resort to the set-intersection-check scheme to obtain
a final estimate \ENDIF }
\end{algorithmic}
\label{algorithm2}
\end{algorithm}





\section{Extension to Antenna Array Receiver}
\label{sec:extension}
In Section \ref{sec:system-model}, we assume the receiver employs
an omni-directional antenna that receives in all directions. In
this section, we extend to the case where both the transmitter and
the receiver have antenna arrays for beam alignment. With a slight
abuse of notation, we let $N_t$ and $N_r$ denote the number of
antennas at the transmitter and the receiver, respectively. The
mmWave channel is characterized by a geometric channel model
\begin{align}
  \boldsymbol{G} = \sum_{p=1}^{P}\alpha_{p}\boldsymbol{a}_r(\theta_{p})\boldsymbol{a}_t^H(\phi_{p})
\end{align}
where $P$ is the number of paths, $\alpha_{p}$ is the complex gain
associated with the $p$th path, $\theta_{p}\in[0,2\pi]$ and
$\phi_{p}\in[0,2\pi]$ are the associated azimuth angle of arrival
(AoA) and angle of departure (AoD), respectively, and
$\boldsymbol{a}_{r}\in\mathbb{C}^{N_r}$
($\boldsymbol{a}_{t}\in\mathbb{C}^{N_t}$) denotes the receiver
(transmitter) array response vector. We assume that the uniform
linear array is used at both the transmitter and receiver. Since
there are only a few paths between the transmitter and the
receiver, the channel matrix in the beam space domain has a sparse
representation
\begin{align}
  \boldsymbol{G} = \boldsymbol{D}_r\boldsymbol{\bar{G}}\boldsymbol{D}_t^H
\end{align}
where $\boldsymbol{D}_r\in\mathbb{C}^{N_r\times N_r}$ and
$\boldsymbol{D}_t\in\mathbb{C}^{N_t\times N_t}$ are the DFT
matrices, and $\boldsymbol{\bar{G}}\in\mathbb{C}^{N_r \times N_t}$
is a sparse matrix. Suppose the transmitter sends a constant
signal $s(t) = 1$ to the receiver. The phaseless measurement
received at the $t$th time instant can be expressed as
\begin{align}
  y(t) &= |\boldsymbol{c}^H(t)\boldsymbol{G}\boldsymbol{b}(t)s(t) + w(t)| \nonumber \\
  &=|\boldsymbol{c}^H(t)\boldsymbol{D}_r\boldsymbol{\bar{G}}\boldsymbol{D}_t^H\boldsymbol{b}(t) + w(t)|
\end{align}
where $\boldsymbol{c}(t)$ denotes the combining vector used at the
receiver. To perform beam alignment, we can let the receiver steer
its beam to a fixed direction over a period of time (or multiple
beams towards different directions if multiple RF chains at the
receiver are available), and let the transmitter send the
codewords devised according to our proposed sparse encoding
scheme. Specifically, the receiver uses a certain column of
$\boldsymbol{D}_r$ as its combining vector, i.e.
$\boldsymbol{c}(t)=\boldsymbol{D}_r[:,i]$, over a period of time,
say $t=1,\ldots,T$. The beamforming vector employed by the
transmitter is the same as discussed in Section
\ref{sec:system-model}, i.e.
$\boldsymbol{b}(t)=\boldsymbol{D}_t\boldsymbol{S}(t)\boldsymbol{f}_{\text{BB}}(t)\triangleq
\boldsymbol{D}_t\boldsymbol{a}(t)$. Thus we have
\begin{align}
  y(t) &= |\boldsymbol{a}^H(t)\boldsymbol{\bar{g}}_i + w^*(t)|
  \quad  t=1,\ldots,T
\end{align}
where $\boldsymbol{\bar{g}}_i$ denotes the $i$th column of
$\boldsymbol{\bar{G}}^H$. We see that the problem is now converted
to the sparse encoding and phaseless decoding problem discussed in
this paper, and our proposed scheme can be used to recover
$|\boldsymbol{\bar{g}}_i|$. After the receiver has scanned all
possible $N_r$ beam directions, we are able to obtain the full
knowledge of $|\boldsymbol{\bar{G}}|$, based on which the best
beamformer-combiner pair can be obtained. Such a beam alignment
scheme has a sample complexity of
$\mathcal{O}(N_r\bar{K}^2/R_{r})$, where $R_r$ represents the
number of RF chains at the receiver, and
$\bar{K}\triangleq\max\{K_1,\ldots,K_{N_r}\}$, with $K_i$ denoting
the number of nonzero entries in the $i$th column of
$\boldsymbol{\bar{G}}^H$, i.e. $\boldsymbol{\bar{g}}_i$. Clearly,
$\bar{K}$ is much smaller than the total number of paths $P$.






\begin{figure*}[!t]
\centering
 \subfigure[Success rates vs. $T$.]{\includegraphics[width=2.3in]{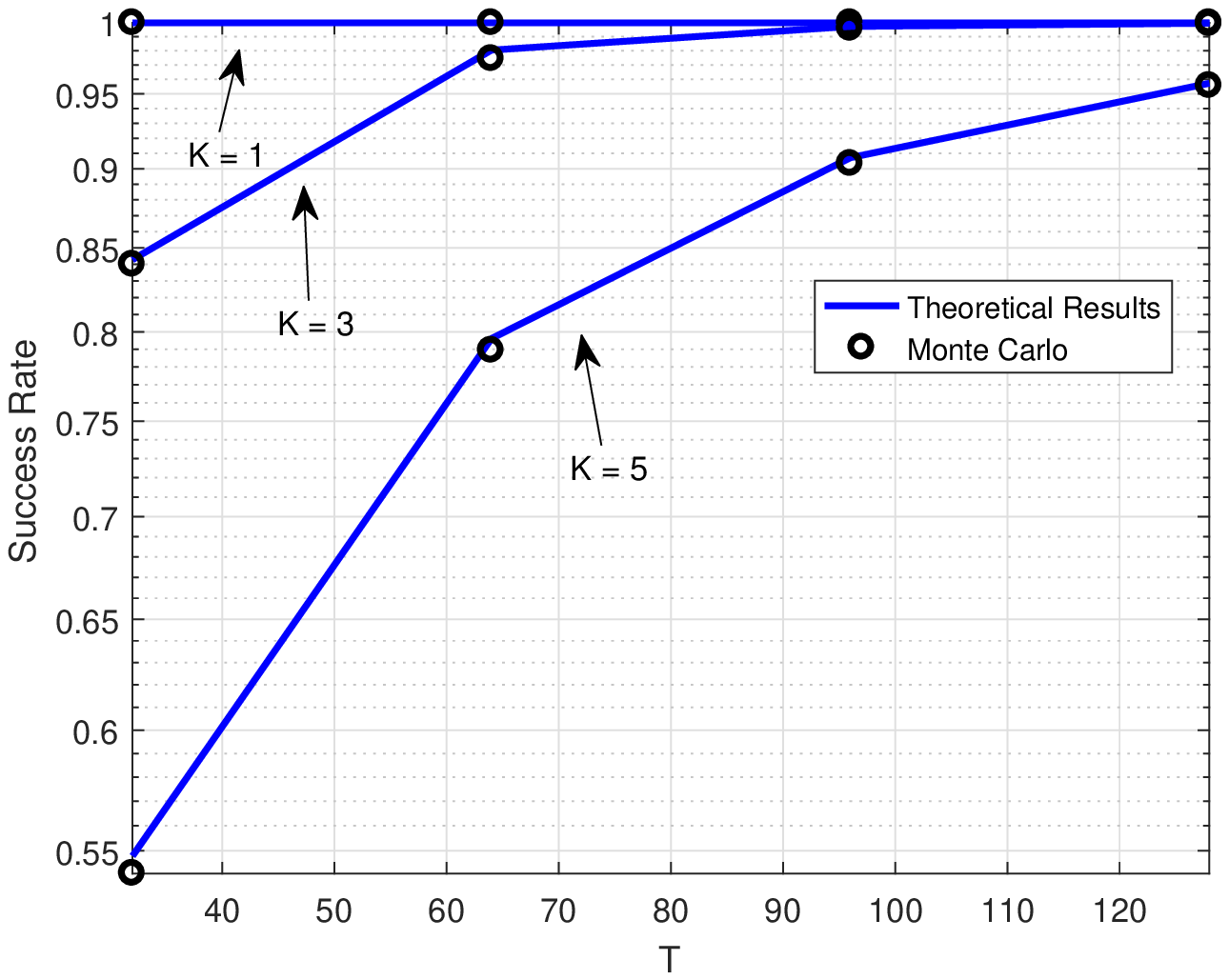}} \hfil
 \subfigure[Success rates vs. $N$.]{\includegraphics[width=2.3in]{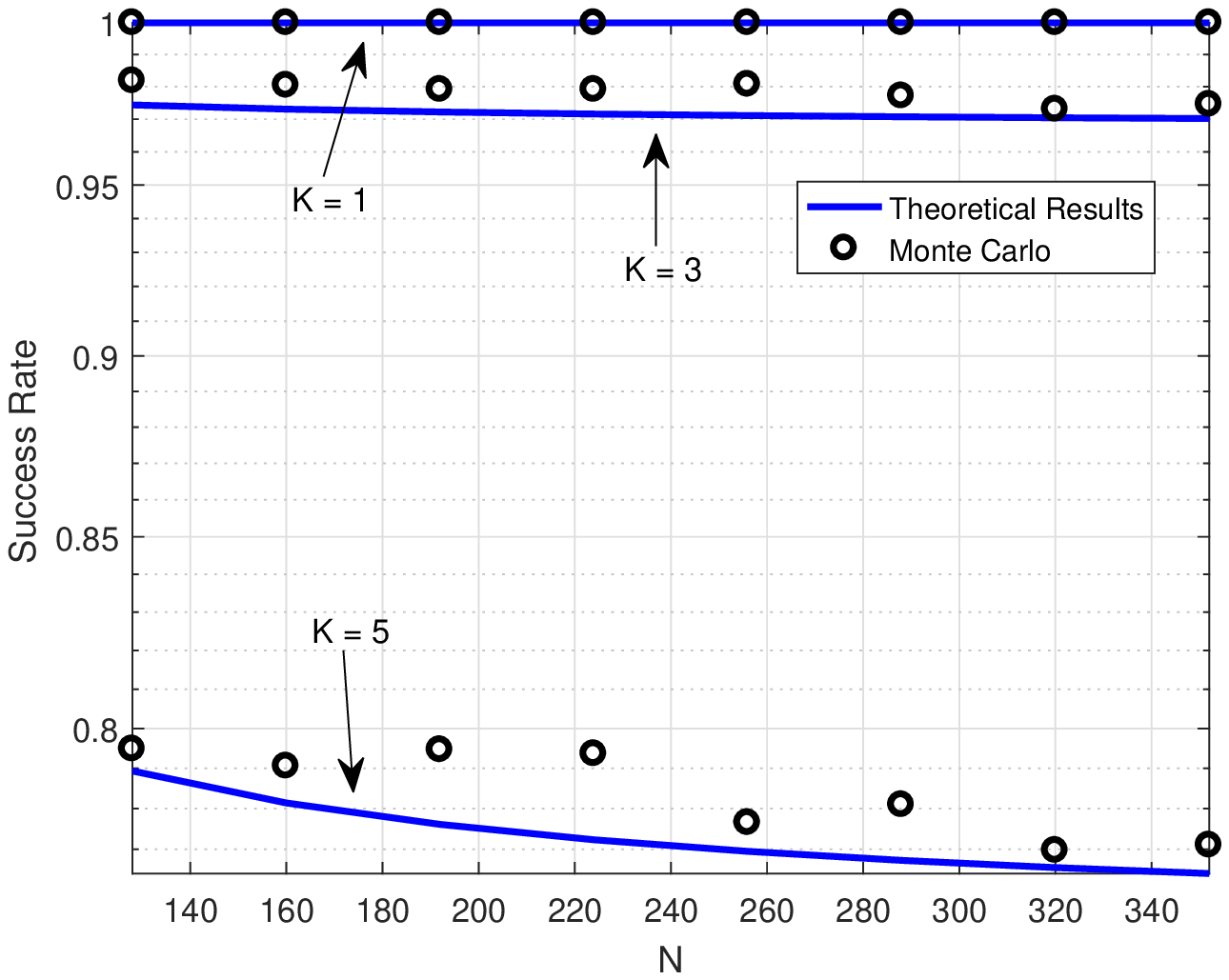}} \hfil
 \subfigure[Success rates vs. $M$.]{\includegraphics[width=2.3in]{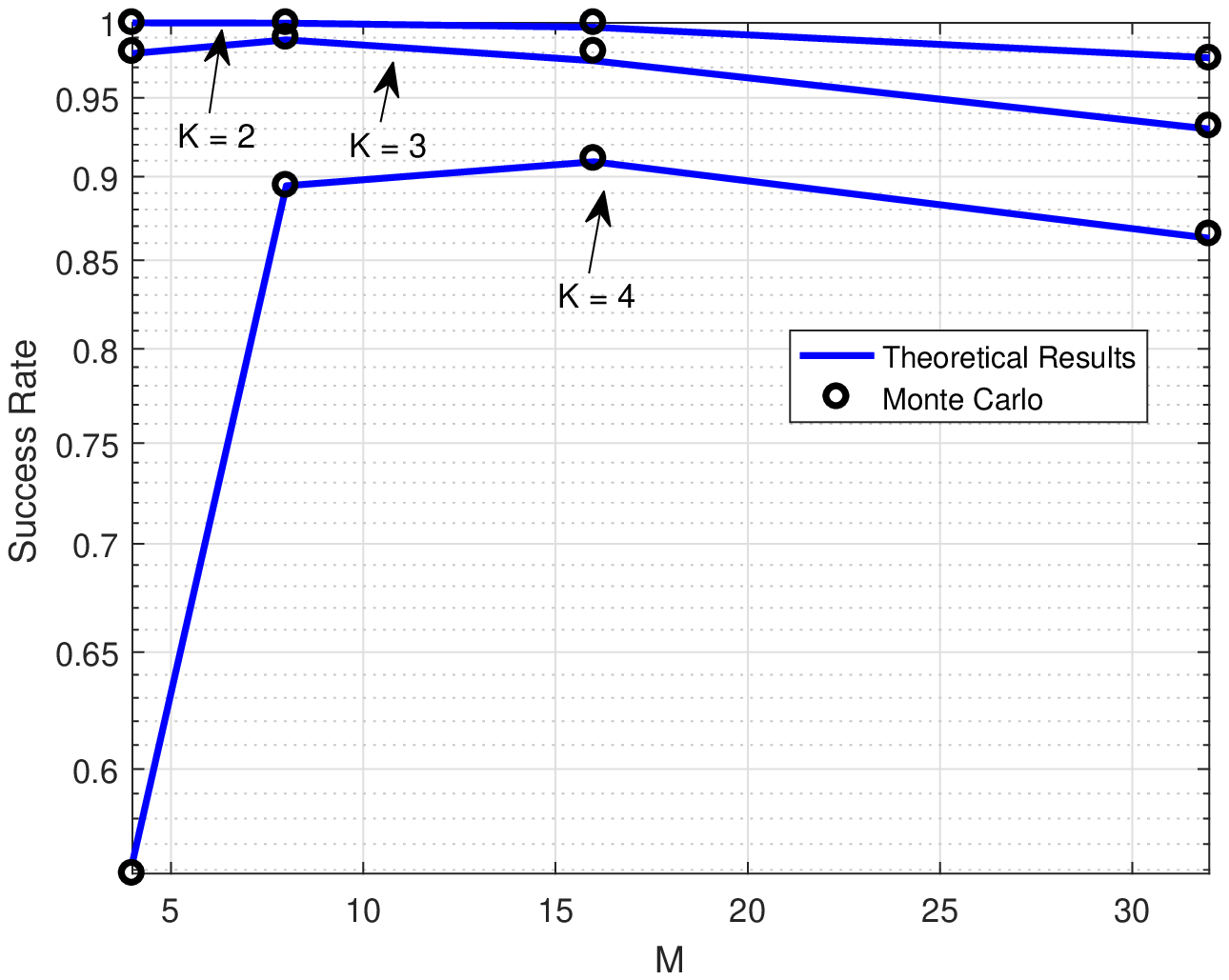}}
\caption{Success rates of our proposed method vs. $T$, $N$, and
$M$ in the noiseless case.} \label{fig1}
\end{figure*}

\begin{figure*}[!t]
\centering \subfigure[NMSEs vs.
$T$.]{\includegraphics[width=3.5in]{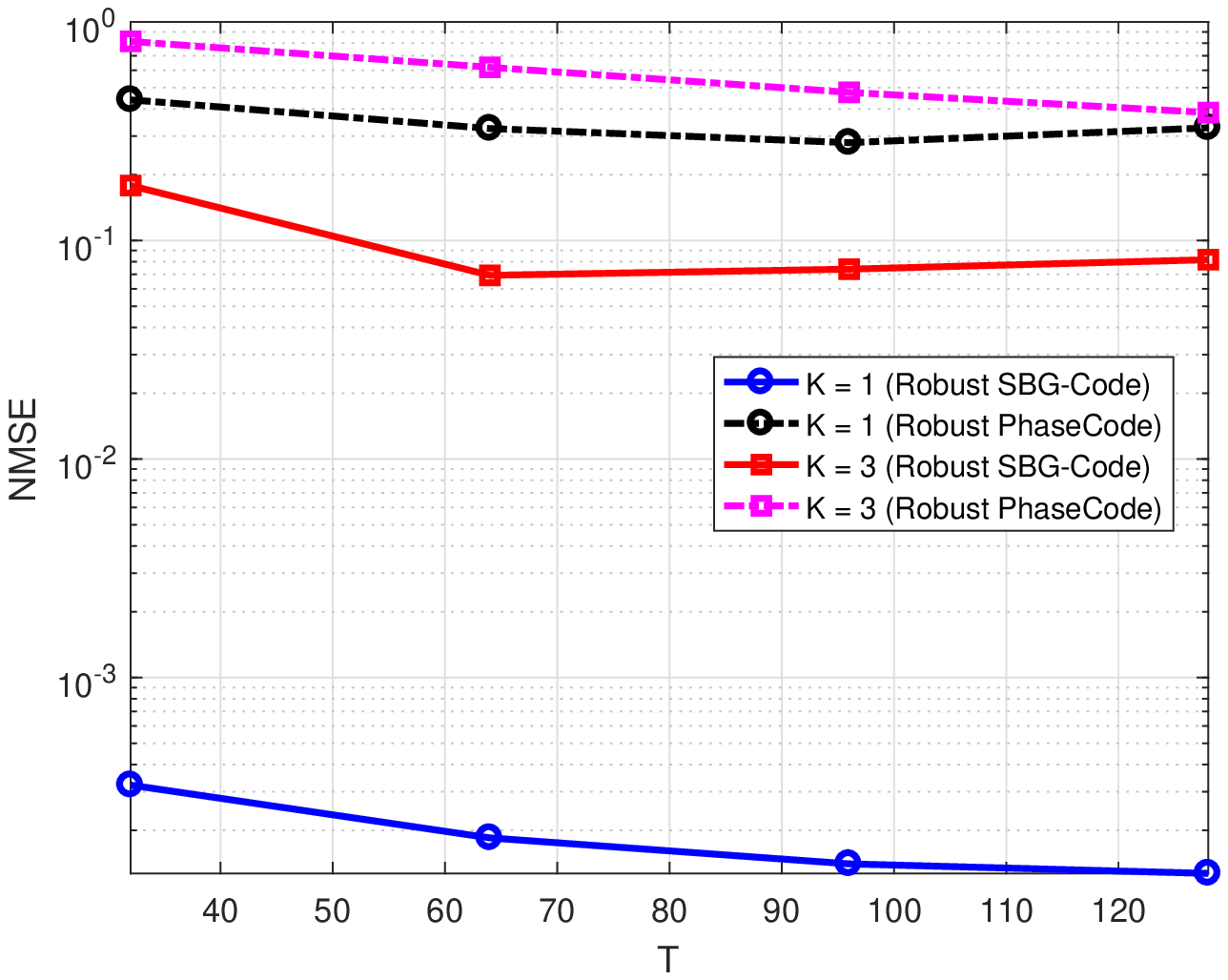}} \hfil
\subfigure[NMSEs vs.
SNR.]{\includegraphics[width=3.5in]{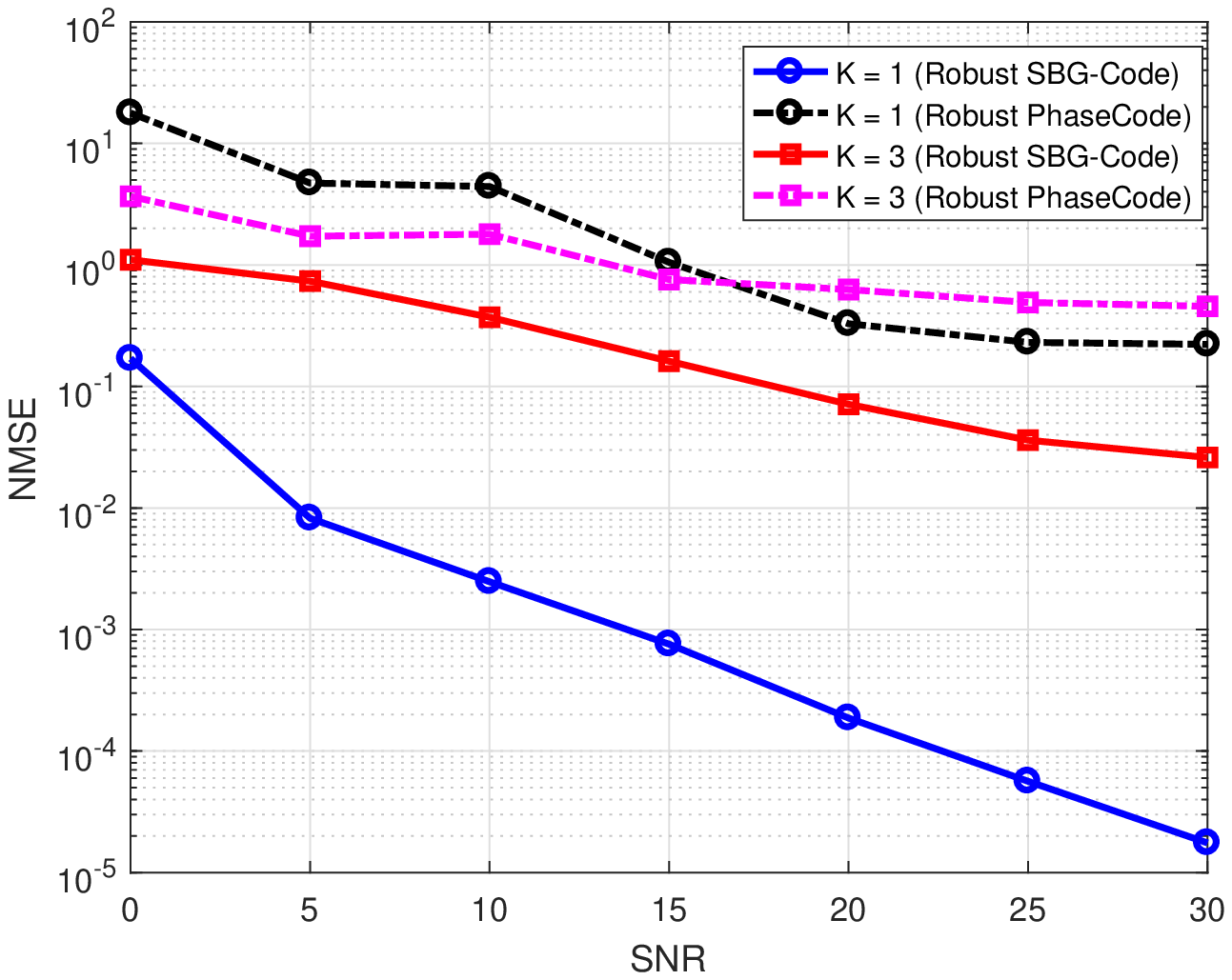}} \caption{NMSEs of
respective algorithms vs. $T$ and SNR.} \label{fig2}
\end{figure*}

\begin{figure*}[!t]
\centering \subfigure[Beamforming gains vs.
$T$.]{\includegraphics[width=3.5in]{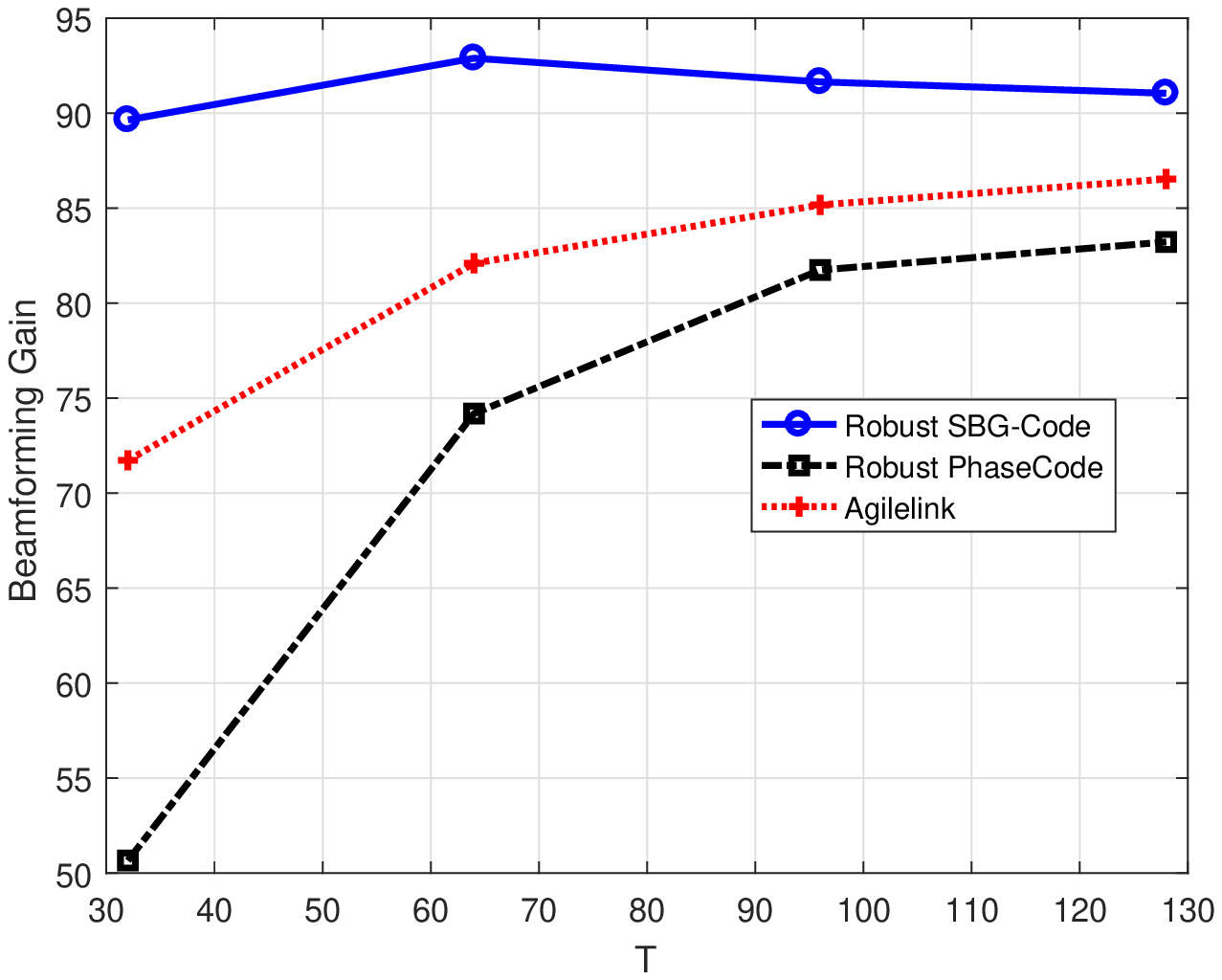}} \hfil
\subfigure[Beamforming gains vs.
SNR.]{\includegraphics[width=3.5in]{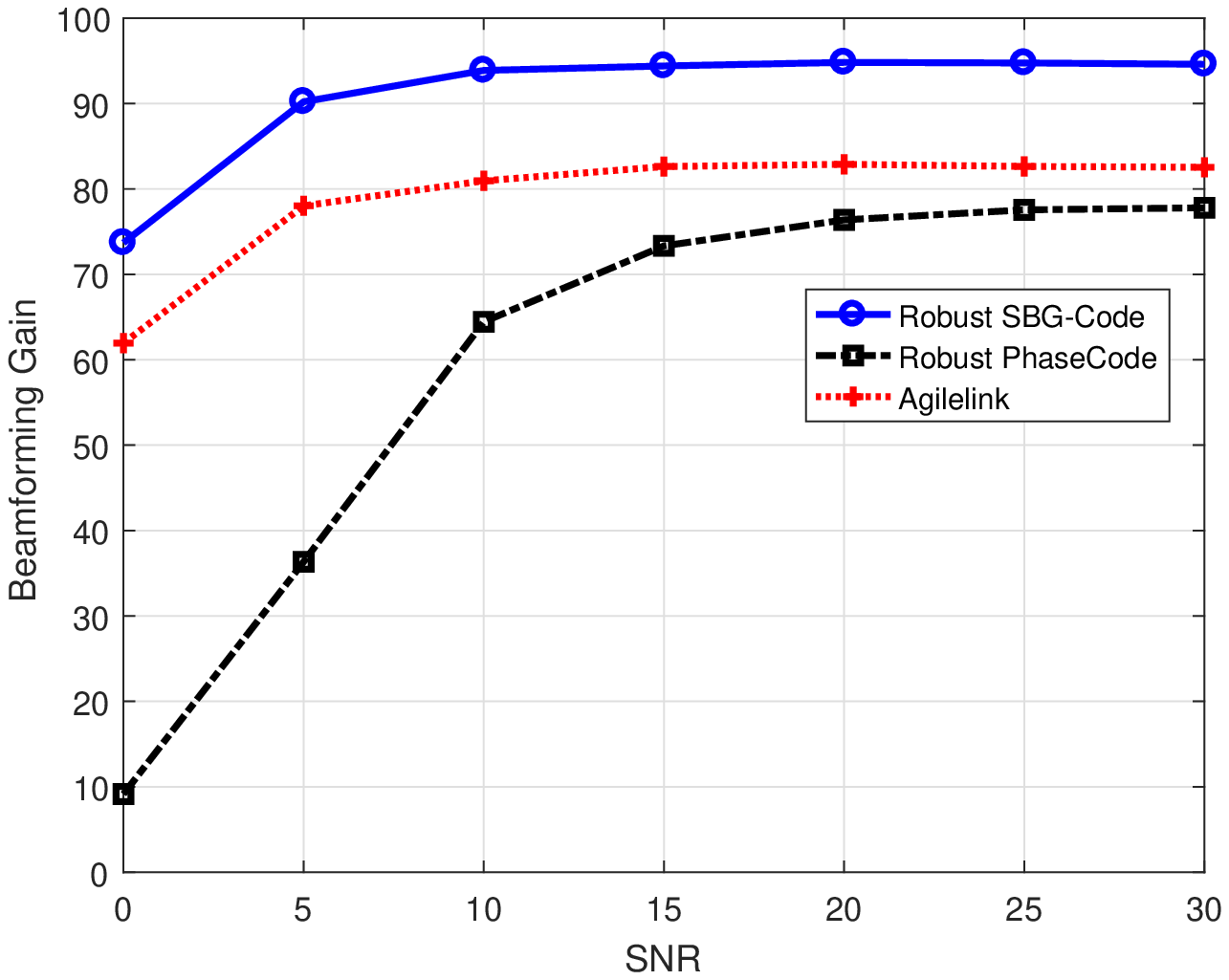}} \caption{Beamforming
gains of respective algorithms vs. $T$ and SNR.} \label{fig3}
\end{figure*}

\section{Simulation Results} \label{sec:simulation-results}
We now present simulation results to illustrate the performance of
our proposed SBG-Code algorithm. In our simulations, the
transmitter employs a ULA with $N$ antennas and $R$ RF chains,
while the receiver uses an omni-directional antenna. The distance
between neighboring antenna elements is assumed to be
$d=\lambda/2$. The mmWave channel $\boldsymbol{h}$ is assumed to
have a form of (\ref{sparse-channel}) with $K$ paths. The nonzero
components of $\boldsymbol{x}$ are assumed to be random variables
following a circularly symmetric complex Gaussian distribution
$\mathcal{CN}(0,1)$, and the locations of nonzero entries of
$\boldsymbol{x}$ are uniformly chosen at random. All the results
are averaged over $10^4$ independent runs. In each run,
$\boldsymbol{x}$ (i.e. $\boldsymbol{h}$) is randomly generated.
The linear function $f(n)=n/N$ is employed to encode the sparse
signal, i.e. the trignometric modulation matrix is given by
(\ref{modulation-matrix-2}), and the estimator (\ref{estimator2})
is used to estimate $\boldsymbol{z}$ in the noiseless case.

We first examine the estimation performance of our proposed
algorithm in the noiseless case. The performance is evaluated via
the success rate, which is computed as the ratio of the number of
successful trials to the total number of independent runs. A trial
is considered successful if
$\|\boldsymbol{\hat{z}}-\boldsymbol{z}\|_2^2/\|\boldsymbol{z}\|_2^2<10^{-8}$,
where $\boldsymbol{\hat{z}}$ denotes the estimate of
$\boldsymbol{z}$. Fig. \ref{fig1}(a) depicts the success rates as
a function of the number of measurements $T=2ML$, where the number
of antennas is set to $N=128$, the number of RF chains is set to
$R=8$, and the number of right nodes in each bipartite graph is
set to $M=16$. In the figure, solid lines represent the
theoretical performance given in (\ref{success-rate}), while the
circle marks represent the performance obtained via the Monte
Carlo experiments. From Fig. \ref{fig1}, we see that our
theoretical result matches the empirical result very well. Also,
when the number of paths $K$ is small, our proposed scheme can
perfectly recover the AoA and the attenuation (in magnitude) of
each path with a decent probability even using a small number of
measurements, say $T=32$, thus achieving a substantial overhead
reduction for beam alignment. Fig. \ref{fig1}(b) plots the success
rates as a function of the dimension of the sparse signal $N$,
where we set $T=64$, $R=8$, and $M=16$. From Fig. \ref{fig1}(b),
we observe that the success rate of our proposed algorithm remains
almost unaltered as $N$ grows. This result corroborates our
theoretical claim that our proposed algorithm has a sample
complexity independent of $N$. It is also interesting to examine
the impact of the choice of the number of right nodes per
bipartite graph, $M$, on the performance of our proposed
algorithm, given the total number of measurements $T$ fixed. Fig.
\ref{fig1}(c) plots the success rates as a function of $M$, where
we set $T=64$ and $N=128$. Note that since the parameter $M$ must
be chosen such that $R\geq \text{floor}(N/M)$, the number of
required RF chains changes as $M$ varies. From Fig. \ref{fig1}(c),
we see that the best performance is achieved when $M\approx K^2$.

Next, we illustrate the estimation performance of our proposed
algorithm in the noisy case. We compare our method with the robust
PhaseCode algorithm. As mentioned earlier in our paper, PhaseCode
uses a single randomly generated bipartite graph to encode the
sparse signal. The resulting measurement matrix $\boldsymbol{A}$
may not satisfy constraint C1. To fulfil the potential of
PhaseCode, we allow the constraint C1 to be violated by PhaseCode.
For our proposed method, the prescribed false alarm probability
used to determine the threshold in the energy detector
(\ref{energy-detector}) is set to $e^{-9/2}\approx0.011$, thus the
threshold is given by $\epsilon=3\sigma$. For a fair comparison,
the beamforming vector $\boldsymbol{b}(t)$ (cf. (\ref{eqn12}))
used in both schemes is normalized to unit norm. The performance
is evaluated via the normalized mean squared error (NMSE)
calculated as
\begin{align}
  \text{NMSE} = E\left[\frac{\|\boldsymbol{\hat{z}}-\boldsymbol{z}\|_2^2}{\|\boldsymbol{z}\|_2^2}\right]
\end{align}
Note that PhaseCode is able to retrieve the complete information
of $\boldsymbol{x}$. But the accuracy of the estimate of
$\boldsymbol{z}=|\boldsymbol{x}|$ is of most concern for beam
alignment. Fig. \ref{fig2}(a) shows the NMSEs of respective
schemes as a function of $T$, where we set $N=128$, $M=16$, and
the SNR is set to $20$dB. Here the SNR is defined as
\begin{align}
\text{SNR}=10\log(\|\boldsymbol{h}\|_2^2/(N\sigma^2))
\end{align}
From Fig. \ref{fig2}(a), we see that our proposed method
outperforms the robust PhaseCode method by a big margin for
difference choices of $K$. The performance improvement is
primarily due to the fact that our proposed method circumvents the
complicated decoding procedure that is needed for PhaseCode and
thus gains substantially improved robustness against noise. Fig.
\ref{fig2}(b) depicts the NMSEs of respective schemes as a
function of SNR, where we set $T=64$ and $M=16$. It can be
observed that our proposed method attains a decent accuracy even
in the low and moderate SNR regimes, whereas the robust PhaseCode
fails in this case.

Lastly, we compare our proposed algorithm with the Agile-Link
\cite{AbariHassanieh16}, a beam steering scheme which also relies
on the magnitude information of measurements for recovery of
signal directions. It should be noted that Alige-Link only
recovers signal directions, but not $\boldsymbol{z}$. The
beamforming gain defined below is used as a metric to evaluate the
performance of respective schemes
\begin{align}
  G_{\text{BF}} = E\left[N|\boldsymbol{a}_t^H(\hat{\theta}_{\text{opt}})\boldsymbol{h}|^2/\|\boldsymbol{h}\|_2^2\right]
\end{align}
in which $\hat{\theta}_{\text{opt}}$ denotes the estimated
direction of path that delivers the maximum energy. For the
Agile-Link, $\hat{\theta}_\text{opt}$ is estimated as the
direction with the highest probability. Fig. \ref{fig3}(a) depicts
the beamforming gains of respective algorithms as a function of
$T$, where we set $N=128$, $K=2$, $M=16$, and
$\text{SNR}=15\text{dB}$. Again, for a fair comparison, the
beamforming vector $\boldsymbol{b}(t)$ used in these schemes is
normalized to unit norm. We see that our proposed method yields a
higher beamforming gain than the Agile-Link and the PhaseCode, and
the performance gap is particularly pronounced when the number of
measurements $T$ is small. This result suggests that our proposed
method can help find a better beam alignment. Fig. \ref{fig3}(b)
plots the beamforming gains of respective algorithms as a function
of SNR, where we set $T=64$ $K=2$, and $M=32$, from which we can
see that our proposed method even renders a decent beamforming
gain in the low SNR regime.

\section{Conclusions} \label{sec:conclusions}
The problem of mmWave beam alignment was examined in this paper.
By exploiting the sparse scattering nature of mmWave channels, we
showed that the problem of beam alignment can be formulated as a
sparse encoding and phaseless decoding problem. A SBG-Code method
was developed to encode the sparse signal and retrieve the support
and magnitude information of the sparse signal from compressive
phaseless measurements. Our analysis revealed that the proposed
method can provably recover the sparse signal with a pre-specified
probability from $\mathcal{O}(K^2)$ phaseless measurements.
Simulation results showed that the proposed scheme renders a
reliable beam alignment even in a low or moderate SNR regime with
very few measurements, and presents a clear advantage over
existing mmWave beam alignment algorithms.

\useRomanappendicesfalse
\appendices

\section{Proof of Proposition \ref{proposition1}} \label{appA}
Before preceding, we first show that the probability that all
right nodes of $G_l$ are either singletons or nulltons is
maximized when each column of $\boldsymbol{H}_l$ has only one
nonzero element, i.e. each left node is connected to only one
right node. Such a fact can be easily verified via an
edge-deletion operation performed on $G_l$. Specifically, for each
left node of $G_l$, if it has more than one edge, that is, it is
connected to more than one right node, then we reserve only one
edge and delete all the other edges. It is clear that after the
edge-deletion operation, the number of singletons and nulltons of
$G_l$ either keeps increased or unchanged. Therefore, the
probability that all right nodes of $G_l$ are either singletons or
nulltons is maximized when each column of $\boldsymbol{H}_l$ has
only one nonzero element. Note than in this case, we have
\begin{align}
  \sum_{m=1}^M r_{m} = rM = N
\end{align}


We now calculate the probability that $G_l$ is an NM-graph when
each left node is connected to only one right node. More
precisely, we divide $N$ left nodes into $M$ disjoint sets, where
the $m$th set consisting of $r_{m}$ left nodes is connected to the
$m$th right node. There are $K$ active left nodes in total. We
need to calculate the probability that each set of left nodes,
denoted as $S_m$, contains at most one active left node. Define
\begin{align}
\mathcal{M}\triangleq\{1,\ldots,M\}
\end{align}
Let $\mathcal{K}\triangleq\{i_1,\ldots,i_K\}$ be a subset of
$\mathcal{M}$ consisting of $K$ elements, and
$\{i_{K+1},\ldots,i_M\}=\mathcal{M}-\mathcal{K}$ be the difference
set between $\mathcal{M}$ and $\mathcal{K}$. It can be easily
verified that the number of ways of dividing $N$ left nodes into
$M$ disjoint sets such that each set $S_m, m\in\mathcal{K}$,
contains only one active left node is given as
\begin{align}
  &K!C_{N-K}^{r_{i_1}-1}C_{N-R_1-K+1}^{r_{i_2}-1}\cdots C_{N-R_{K-1}-1}^{r_{i_K}-1}
  C_{N-R_K}^{r_{i_{K+1}}}\cdots C_{N-R_{M-1}}^{r_{i_M}} \nonumber \\
  &=\frac{K!(N-K)!}{\prod_{t=1}^K (r_{i_t}-1)!\prod_{t=K+1}^M r_{i_t}!}
\end{align}
where $R_n \triangleq \sum_{t=1}^{n}r_{i_t}$. Thus, the number of
ways of dividing $N$ left nodes into $M$ disjoint sets such that
each set contains at most one active left node is given by
\begin{align}
n_1 \triangleq \sum_{\{i_1,\ldots,i_{K}\}\subseteq\mathcal{M}}
\frac{K!(N-K)!}{\prod_{t=1}^K (r_{i_t}-1)!\prod_{t=K+1}^M
r_{i_t}!}
\end{align}
On the other hand, the total number of ways of assigning $N$ left
nodes to $M$ disjoint sets is given as
\begin{align}
  n_{2} \triangleq C_{N}^{r_{1}}C_{N-r_1}^{r_{2}}\cdots C_{N-\sum_{i=1}^{M-1}r_i}^{r_{M}} = \frac{N!}{\prod_{i=1}^M r_{i}!}
\end{align}
Therefore the probability that $G_l$ is an NM graph can be
calculated as
\begin{align}
P(\text{$G_l$ is an NM-graph}) =\frac{n_1}{n_2}
=\frac{\eta(K)}{C_N^K}
\end{align}
where
\begin{align}
\eta(K)\triangleq\sum_{\{i_1,\ldots,i_{K}\}\subseteq\mathcal{M}}\bigg(\prod_{t=1}^{K}
r_{i_t}\bigg)
\end{align}

Next, we prove
\begin{align}
\eta(K)\leq r^K C_M^K \label{inequality2}
\end{align}
holds for all $1\leq K\leq M$ when $\sum_{i=1}^M r_{i} = rM$. The
inequality (\ref{inequality2}) is proved by mathematical
induction. First, we prove the base case: $K=1$. It is easy to
verify that
\begin{align}
  \eta(1) = \sum_{i=1}^M r_{i}= rM = r C_M^1
\end{align}
We then proceed to the inductive step. Suppose the following
inequality holds for $K'-1$
\begin{align}
  \eta(K'-1) \leq r^{K'-1} C_M^{K'-1}
  \label{k'-1}
\end{align}
We need to prove
\begin{align}
  \eta(K') \leq r^{K'} C_{M}^{K'}
  \label{th1-1}
\end{align}
To this goal, we multiply both sides of (\ref{k'-1}) by
$\sum_{i=1}^M r_{i}$, which yields
\begin{align}
  \eta(K'-1) \bigg(\sum_{i=1}^M r_{i}\bigg)
  \leq& \ r^{K'-1} C_M^{K'-1} Mr \nonumber \\
  =&\ Mr^{K'} C_M^{K'-1}
  \label{th1-2}
\end{align}
The left-hand side of (\ref{th1-2}) can be further written as
\begin{align}
  &\eta(K'-1)\bigg(\sum_{i=1}^M r_{i}\bigg)   \nonumber \\
  =\ &\sum_{i=1}^M\sum_{\{i_1,\ldots,i_{K'-2}\}\subseteq\mathcal{M}-\{i\}}
  r_{i}^2r_{i_1}\cdots r_{i_{K'-2}} + K'\eta(K')  \nonumber \\
  = &\frac{1}{M-K'+1}\sum_{\{i_1,\ldots,i_{K'}\}\subseteq\mathcal{M}}
  \left(\prod_{t=1}^{K'} r_{i_t} \sum_{\substack{j,k\in\{1,\ldots,K'\}\\j\neq k}}
  \frac{r_{i_k}}{r_{i_j}}\right)\nonumber \\
  &+ K'\eta(K')
  \label{th1-3}
\end{align}
For any $\{i_1,\ldots,i_{K'}\}\subseteq\mathcal{M}$, using the
inequality of arithmetic and geometric means (also referred to as
the AM-GM inequality), we have
\begin{align}
  \sum_{\substack{j,k\in\{1,\ldots,K'\}\\j\neq k}}
  \frac{r_{i_k}}{r_{i_j}}
  \geq \ &K'(K'-1)\left(\prod_{\substack{j,k\in\{1,\ldots,K'\}\\j\neq k}}
  \frac{r_{i_k}}{r_{i_j}}\right)^{\frac{1}{K'(K'-1)}}  \nonumber \\
  = \ &K'(K'-1)
\end{align}
in which the inequality becomes an equality if and only if
$r_{i_1}=\cdots =r_{i_{K'}}$. Hence, we have
\begin{align}
  &\sum_{\{i_1,\ldots,i_{K'}\}\subseteq\mathcal{M}}\left( \prod_{t=1}^{K'} r_{i_t}
  \sum_{\substack{j,t\in\{1,\ldots,K'\}\\j\neq t}}
  \frac{r_{i_t}}{r_{i_j}}\right)
  \geq K'(K'-1)\eta(K')
  \label{th1-5}
\end{align}
in which the inequality (\ref{th1-5}) becomes equality if and only
if $r_{1}=\cdots =r_{M}=r$. Combining (\ref{th1-2}), (\ref{th1-3})
and (\ref{th1-5}), we arrive at
\begin{align}
  &Mr^{K'}C_M^{K'-1}
  \geq\eta(K'-1) \bigg(\sum_{i=1}^M r_{i}\bigg) \nonumber \\
  \geq\ &\left(\frac{K'(K'-1)}{M-K'+1}+K'\right)\eta(K')
  =\frac{MK'}{M-K'+1}\eta(K') \label{appA:eqn1}
\end{align}
From (\ref{appA:eqn1}), we have
\begin{align}
  \eta(K')
  \leq\frac{M-K'+1}{MK'}Mr^{K'} C_M^{K'-1}
  =r^{K'} C_M^{K'}
\end{align}
Thus the inductive step is proved. This completes our proof.

\section{Proof of Theorem \ref{theorem1}} \label{appB}
According to our proposed algorithm, we see that the support and
magnitude information of $\boldsymbol{x}$ can be perfectly
recovered when there is at least one NM-graph in all bipartite
graphs $\{G_l\}_{l=1}^L$. Therefore, the probability that our
proposed algorithm succeeds to recover the support and magnitude
information of $\boldsymbol{x}$ equals the probability that there
is at least one NM-graph in $\{G_l\}_{l=1}^L$, which is equivalent
to
\begin{align}
  p = 1-(1-P(\text{$G_l$ is an NM-graph}))^L = 1-(1-\lambda)^L
\end{align}

\bibliography{newbib}
\bibliographystyle{IEEEtran}

\end{document}